\documentclass[final, 3p]{elsarticle}
\usepackage{hyperref}
\usepackage{textcase}
\usepackage{setspace}
\usepackage{mathtools}
\usepackage{array}
\usepackage{tikz}
\usepackage[section]{placeins}
\usepackage{bm}
\usepackage{bbm}
\usepackage{subcaption}
\usepackage{caption}
\usepackage{float}
\usepackage{amsmath}
\usepackage{graphicx}
\usepackage{epstopdf}
\usepackage{enumitem}
\usepackage{amssymb}
\usepackage[symbol]{footmisc}
\usepackage{amssymb,amsmath}
\usepackage{amsgen,amsfonts,amsbsy,amsthm}
\usepackage{latexsym}
\usepackage{times}
\usepackage{algorithm}
\usepackage[noend]{algpseudocode}
\makeatletter
\def\BState{\State\hskip-\ALG@thistlm}
\makeatother
\newcounter{defcounter}
\newcommand{\spn}[1]{\mbox{\texttt{span}}\left(#1\right)}

\newcommand{\bvec}[1]{\boldsymbol{#1}}
\newcommand{\real}{\mathbb{R}}

\newcommand{\opnorm}[2]{\left\|#1\right\|_{#2}}
\newcommand{\norm}[1]{\left\|#1\right\|_{2}}
\newcommand{\abs}[1]{\left|#1\right|}

\newcommand{\conv}[1]{\texttt{conv}\left(#1\right)}

\newcommand{\expect}[2]{\mathbb{E}_{#1}\left(#2\right)}
\newcommand{\indicator}[1]{\mathbbm{1}_{\left\{#1\right\}}}
\newtheorem*{remark}{Remark}
\newtheorem{thm}{Theorem}[section]
\newtheorem{lem}[thm]{Lemma}
\newtheorem{cor}{Corollary}[thm]
\theoremstyle{definition}
\newtheorem{define}{Definition}[section]

\allowdisplaybreaks
\journal{Elsevier Signal Processing}
\begin{document}
	\setlength\abovedisplayskip{0pt}
	\setlength\belowdisplayskip{0pt}
	\begin{frontmatter}	
		\title{Sparse Recovery Analysis of Generalized $J$-Minimization with Results for Sparsity Promoting Functions with Monotonic Elasticity}
		\author{Samrat Mukhopadhyay\corref{cor}}
		\address{Department of Electronics and Electrical Communication Engineering, Indian Institute of Technology, Kharagpur, INDIA}
		\ead{samratphysics@gmail.com}
		%% \ead[url]{home page}
		%\fntext[label2]{1}
		\cortext[cor]{Corresponding Author}
		%	, {\ttEmail address of Second Author} (Second Author), {\tt Email addressof Third Author} (Third Author)}
\begin{abstract}
	In this paper we theoretically study exact recovery of sparse vectors from compressed measurements by minimizing a general nonconvex function that can be decomposed into the sum of single variable functions belonging to a class of smooth nonconvex sparsity promoting functions. Null space property (NSP) and restricted isometry property (RIP) are used as key theoretical tools. The notion of \emph{scale function} associated to a sparsity promoting function is introduced to generalize the state-of-the-art analysis technique of the $l_p$ minimization problem. The analysis is used to derive an upper bound on the null space constant (NSC) associated to this general nonconvex minimization problem, which is further utilized to derive sufficient conditions for exact recovery as upper bounds on the restricted isometry constant (RIC), as well as bounds on optimal sparsity $K$ for which exact recovery occurs. The derived bounds are explicitly calculated when the sparsity promoting function $f$ under consideration possesses the property that the associated \emph{elasticity function}, defined as, $\psi(x)=\frac{xdf(x)/dx}{f(x)}$, is monotonic in nature. Numerical simulations are carried out to verify the efficacy of the bounds and interesting conclusions are drawn about the comparative performances of different sparsity promoting functions for the problem of $1$-sparse signal recovery. 
%	Furthermore, the bounds on NSC are numerically compared with the ones derived from actual NSC for an example recovery problem and the two are compared for several sparsity promoting functions.
	\end{abstract}
\end{frontmatter}
\section{Introduction}
\label{sec:intro}
\subsection{Overview}
\label{sec:overview}
Compressed sensing~\cite{candes_decoding_2005,candes2006robust,donoho2006compressed} is a highly successful theoretical tool in the signal processing community that aims to recover a \emph{sparse} vector $\bvec{x}\in \real^N$ with no more than $K (K\ll N)$ nonzero entries from a small number of linear measurements, $\bvec{y}=\bvec{\Phi x}$, where $\bvec{y}\in \real^M$ with $M<N$. The canonical formulation of the compressed sensing problem is in terms of the following optimization problem: \begin{align}
\label{eq:l0-minimization-problem}
P_0: & \min_{\bvec{x}} \opnorm{\bvec{x}}{0},\ \mathrm{s.t.}\ \bvec{y}=\bvec{\Phi x},
\end{align}
where $\opnorm{\bvec{x}}{0}$ is the $l_0$ ``norm'' of $\bvec{x}$, which counts the number of nonzero entries of $\bvec{x}$. The problem $P_0$~\eqref{eq:l0-minimization-problem}, however, is a combinatorial minimization problem which becomes computationally intractable for even moderate problem size $N$~\cite{candes2006robust}. Candes~\emph{et al}~\cite{candes2006robust} studied the following convex relaxation of the problem $P_0$~\eqref{eq:l0-minimization-problem}:\begin{align}
\label{eq:l1-minimization-problem}
P_1: & \min_{\bvec{x}}\opnorm{\bvec{x}}{1},\ \mathrm{s.t.}\ \bvec{y}=\bvec{\Phi x}.
\end{align}
It was proved in~\cite{candes2006robust} that if $\bvec{x}$ is $K$-sparse, solving the problem $P_1$~\eqref{eq:l1-minimization-problem} using $\mathcal{O}(K\ln N)$ random Fourier measurements will yield the same minimizer as that of $P_0$~\eqref{eq:l0-minimization-problem} with probability $1 - \mathcal{O}(N^{-M})$. Since then, a plethora of efficient techniques have been developed to solve the problem $P_1$~\eqref{eq:l1-minimization-problem}, ranging from an interior point method~\cite{kim2007interior}, a modified homotopy method called the least angle regression (LARS)~\cite{donoho2008fast}, a fast iterative shrinkage-thresholding algorithm (FISTA)~\cite{beck2009fast}, an iteratively reweighted least squares (IRLS) method~\cite{daubechies2010iteratively} and alternating directions method of multipliers (ADMM)~\cite{boyd2011distributed,combettes2011proximal}, to name a few. Because of the convex nature of the problem $P_1$~\eqref{eq:l1-minimization-problem}, all these methods are well-analyzed and have provable convergence guarantees that make them attractive for the purpose of practical usage.

Lately, there has been an increasingly growing interest in studying an alternative formulation of the original problem $P_0$~\eqref{eq:l0-minimization-problem} as below: \begin{align}
\label{eq:J-minimization-problem}
P_f: & \min_{\bvec{x}\in \real^N} J(\bvec{x}),\ \mathrm{s.t.}\ \bvec{y}=\bvec{\Phi x},
\end{align}
where $J:\real^N\to \real$ can be decomposed as $J(\bvec{x})=\sum_{i=1}^N f(x_i)$, where $f:\real\to \real$, is a \emph{sparsity promoting function}~\cite{gribonval2007highly,chen2014convergence}, typically chosen to be a non-negative increasing piece-wise concave function. For different $f$, the performance of this so called \emph{$J-$minimization} has been studied by numerous researchers~\cite{chartrand2008restricted,ba2013convergence,chen2014convergence,jeong2014can,chen2015null,chen2016square-root-lasso,yukawa2016regularized,woodworth2016compressed}. Although the problem $P_f$ is highly non-convex because of the piece-wise concave nature of $f$, it has been proved and demonstrated in these studies that the sparse recovery performance of this non-convex minimization approach outperforms that of the $P_1$~\eqref{eq:l1-minimization-problem} problem, either in terms of faster convergence of the solution methods, or in the requirement of smaller number of measurements to achieve the same level of probability of perfect recovery (when using random measurement matrices, typically with i.i.d. Gaussian entries). In most of these analyses the main tools used are the null space property (NSP) and the restricted isometry property (RIP) associated with the sensing (or measurement) matrix $\bvec{\Phi}$.
\subsection{Relevant literature on $J$-minimization, NSP and RIP}
\label{sec:relevant-literature}
The paper~\cite{gribonval2007highly} first proposed the NSP and the null space constant (NSC) (and the $J-$NSC for problem $P_f$~\eqref{eq:J-minimization-problem}) and first established that to recover a $K-$sparse vector exactly by solving $P_f$~\eqref{eq:J-minimization-problem}, the $J-$NSC has to be strictly less than unity. Also, it was established in~\cite{gribonval2007highly} that for any sparsity promoting function $f$, the corresponding $J-$NSC is greater than or equal to the $l_0-$NSC for $P_0$~\eqref{eq:l0-minimization-problem} and less than or equal to the one for $P_1$~\eqref{eq:l1-minimization-problem}. It was further proved in~\cite{chen2014convergence} that the $J-$NSC for the problem $P_f$~\eqref{eq:J-minimization-problem} is identical to the $l_1-$NSC for $P_1$~\eqref{eq:l1-minimization-problem} when the function $f$ satisfies a certain weak convexity property~\cite{chen2014convergence}. However, many functions, for example $f(x)=\abs{x}^p$, or the continuous mixed norm~\cite{zayyani2014continuous} (which technically is a quasi-norm) do not satisfy such weak convexity property and it is not clear whether this conclusion holds true for the corresponding $J-$NSC.

The paper~\cite{candes_decoding_2005} defined the notion of RIP and the restricted isometry constant (RIC) of order $K$, denoted by $\delta_K$, and proved that the condition $\delta_K+\delta_{2K}+\delta_{3K}<1$ is sufficient for exact recovery of a $K-$sparse signal by solving the problem $P_1$~\eqref{eq:l1-minimization-problem}. Following up this research effort, numerous other works~\cite{cai2014sparse,wen2015stable,zhang2017optimal} have adopted more sophisticated arguments to propose increasingly improved bounds on $\delta_{2K}$ for the recovery of a $K-$sparse vector by solving the problems $P_p$(i.e. $P_f$~\eqref{eq:J-minimization-problem} with $f(x)=\abs{x}^p$), for $p\in(0,1]$. 
Recently, the work~\cite{gu2018sparse-lp-bound} has addressed the issue of both determining an upper bound on the $J-$NSC for the problem $P_p, p\in(0,1]$, with explicit dependence on $\delta_{2K},K,p$, as well as finding upper bounds on the RIC $\delta_{2K}$ that ensures exact recovery of a $K$-sparse signal by solving the problem $P_p$. Furthermore,~\cite{gu2018sparse-lp-bound} has produced the sharpest of all the RIC bounds derived so far. However, in recent years, many researchers have proposed and worked with many other non-convex sparsity promoting functions, for example, the concave exponential $f(x) = 1 - e^{-\abs{x}^p},\ p\in[0,1]$~\cite{bradley1998feature,fung2002minimal,tipping-etal2003l0}, the Lorentzian function $f(x)=\ln(1+\abs{x}^p)$~\cite{carrillo2010robust,carrillo2013lorentzian}, and the continuous mixed norm $f(x)=\expect{p}{\abs{x}^p}$~\cite{zayyani2014continuous}, where the expectation is taken with respect to a probability measure defined on the Borel sets of $[0,1]$ for the random variable $p$. To the best of our knowledge, no effort seems to have been spent on either deriving an upper bound on the $J-$NSC or bounds on $\delta_{2K}$, with explicit dependence on $K$, for such sparsity promoting functions that are neither $l_p$ nor weakly convex.

It is the goal of this paper to investigate the following questions:
\begin{enumerate}
\item Can the arguments of~\cite{gu2018sparse-lp-bound} be extended to the case of general sparsity promoting functions to derive bounds on NSC for the problem $P_f$~\eqref{eq:J-minimization-problem}?
\item Can a bound on $\delta_{2K}$ be derived that ensures exact recovery of a $K$-sparse signal by solving the problem $P_f$~\eqref{eq:J-minimization-problem} for general sparsity promoting functions $f$?
\item For a given sparsity promoting function $f$, can an upper bound for the largest $K$ of exactly recoverable $K$ sparse vector by solving the problem $P_f$~\eqref{eq:J-minimization-problem} be found? 
\end{enumerate}
\subsection{Contribution and Organization}
\label{sec:contribution-organization}
In this paper, we answer the above questions related to the performance of sparse recovery generalized nonconvex minimization. First, after describing the notations adopted in the paper as well as some Lemmas and definitions that are used throughout the rest of the paper in Section~\ref{sec:prelimiaries}, we present the main technical tool in Section~\ref{sec:upper-bound-F-NSP}, where we extend the techniques of~\cite{gu2018sparse-lp-bound} to general sparsity promoting functions. Second, we use this tool to present an upper bound on the null space constant associated with the generalized problem $P_f$ in Section~\ref{sec:upper-bound-F-NSP}. In particular, we show that computing such bound is straight forward for functions which have an associated monotonic \emph{elasticity function}, defined as $\psi(x)=\frac{xf'(x)}{f(x)}$, where $f'(x)=\frac{df(x)}{dx}$. Third, we use this result to find explicit upper bounds on the NSC in Section~\ref{sec:upper-bound-F-NSP} and on the RIC and the maximum recoverable sparsity in Section~\ref{sec:recovery-conditions-ric-p-k} for the problem $P_f$ specialized to three popular functions $f(x)=\ln(1+\abs{x}^p),\ 1-e^{-\abs{x}^p}$ and,$\int_0^1 \abs{x}^pd\nu$ (for a general probability measure $\nu$), and discuss several of their implications.
%
%:\begin{enumerate}
%	\item  with certain smoothness properties, which results in a bound on the NSC for the general problem $P_f$.
%	\item  We present 
%	\item We find upper bound on $\delta_{2K}$ for the problem $P_f$ specialized for three popular functions $f(x)=\ln(1+\abs{x}^p),\ 1-e^{-\abs{x}^p}$ and,$\int_0^1 \abs{x}^pd\nu$, for a general probability measure $\nu$.
%	\item We find bounds on the largest recoverable $K,p$ for the above three functions, by requiring that these derived bounds are strictly less than unity.
%	\item  We numerically demonstrate that there are sparse recovery problems where the $l_p$ function fails while some other sparsity promoting functions succeed.
%\end{enumerate}
%1) 
%2)\\
%3) \\
%4) \\ 
%5)
%
%The rest of the paper is organized as below: .  describes the main technical tool as well as the main result of the paper. Section~\ref{sec:recovery-conditions-ric-p-k} applies the results in Section~\ref{sec:upper-bound-F-NSP} to a few special sparsity promoting functions and find concrete bounds on the RIC and NSC and the maximum recoverable sparsity corresponding to these functions 
In Section~\ref{sec:numerical-results} numerical experiments are performed to study exact recovery of $1-$sparse vectors by solving $P_f$ using different sparsity promoting functions and several interesting observations about comparative recovery performances are illustrated. Furthermore, comparative study of the actual NSC and the derived bounds are numerically studied and the efficacy of the proposed bounds are illustrated for several sparsity promoting functions. 
%The problem $P_p$ (which is the problem $P_f$ with $f(x)=\abs{x}^p,\ p\in (0,1]$) has received a lot of attention and numerous researchers have thoroughly studied its performances using the associated NSP~\cite{foucart2013mathematical,chen2015null,wang2015lp} and RIP~\cite{cai2014sparse,wen2015stable,zhang2017optimal,candes_decoding_2005}. It is proved in   
%
\section{Preliminaries}
\label{sec:prelimiaries}
In this section we describe the notations used throughout the paper and describe some relevant results and introduce some definitions and lemmas regarding the sparsity promoting function, that will be useful in the subsequent analysis in the paper. 
\subsection{Notations}
\label{sec:notations}
The following notations have been used throughout the paper : `$\top$' in
superscript indicates matrix / vector transpose, $[N]$ denotes the set of indices $\{1,2,\cdots,\
N\}$, $x_i$ denotes the $i^\mathrm{th}$ entry of $\bvec{x}$, and $\bm{\phi}_i$ denotes the $i$ th column of $\bvec{\Phi}$. 
%For any two vectors $\bvec{u}, \bvec{v}\in \real^n$, $\inprod{\bvec{u}}{\bvec{v}} = \bvec{u}^t \bvec{v}$.
%\item $\|\cdot\|_p$ ($p>0$) denotes the $l^p$ norm of a vector $\bvec{v}$, i.e., $\|\bvec{v}\|_p=\left(\sum_{i=1}^n |v_i|^p\right)^{1/p}$.
For any $S$, $S^C$ denotes the complement of the set. $\bvec{x}_S$ denotes a vector in $\real^N$, such that $[\bvec{x}_S]_i = x_i$, if $i\in S$, and $[\bvec{x}_S]_i = 0$ if $i\notin S$. Similarly, $\bvec{\Phi}_S$ denotes
the submatrix of $\bvec{\Phi}$ formed with the columns of
$\bvec{\Phi}$ having column numbers given by the index set $S$. The null space of a matrix $\bvec{\Phi}$ is denoted by $\mathcal{N}(\bvec{\Phi})$. The generalized $f-$mean of a vector $\bvec{x}\in \real^N$ is defined by $f^{-1}\left(\frac{1}{N}\sum_{i=1}^N f(x_i)\right).$ For any two real numbers $a,b$, we denote $a\wedge b = \min\{a,b\}$, and $a\vee b = \max\{a,b\}$. The set $\real_+$ denotes the set of all non-negative real numbers. For any function $f$ of a single real variable $f'$ denotes its derivative. The notation $\expect{p}{\cdot}$ is used to denote expectation with respect to the random variable $p$.
\subsection{Definitions}
\label{sec:prelimianries-definitions}
%
%The function $f$ is assumed to satisfy the following properties in Definition~\ref{def:f-properties}~\cite{gribonval2007highly,chen2014convergence,shen2018nonconvex}
\begin{define}
\label{def:f-properties}
A function $f$ is called a \emph{sparsity promoting function} if it satisfies the following properties:
\begin{enumerate}
\item $f(0)=0$, and is an even function.
\item $f$ is strictly increasing in $[0,\infty)$ and is continuous at $0$.
\item $f$ is differentiable and concave in $(0,\infty)$.
\end{enumerate}
\end{define}
Properties similar to the ones in Definition~\ref{def:f-properties} first appeared in~\cite{gribonval2007highly}, where $f$ was neither required to be continuous at $0$ nor was required to be differentiable in $(0,\infty)$, and the requirement of concavity of $f$ in $(0,\infty)$ was replaced by the weaker requirement that $f(x)/x$ is non-increasing in $(0,\infty)$. However, in the current Definition~\ref{def:f-properties} a bit more smoothness is assumed, without which the analysis might become riddled with unnecessary technical difficulties arising due to non-differentiability issues. 
\begin{define}[Definition~$2$ of~\cite{chen2014convergence}]
\label{def:J-NSC-definition}
For a sparsity promoting function $f$ let the function $J:\real^N\to \real$ be defined as $J(\bvec{x})=\sum_{j=1}^N f(x_j)$ for all $\bvec{x}\in \real^N$. The $J$ null space constant ($J$-NSC) $\gamma(J,\bvec{\Phi},K)$ of matrix $\bvec{\Phi}$ is the smallest number $\gamma>0$ such that \begin{align}
\label{eq:J-NSC-inequality}
J(\bvec{z}_{S})\le \gamma J(\bvec{z}_{S^C}),
\end{align} 
for any $\bvec{z}\in \mathcal{N}(\bvec{\Phi})$ and any subset $S\subset [N]$ such that $\abs{S}\le K$.
\end{define}
\begin{define}
	\label{def:restricted-isometry-property}
	For a given sparsity order $K$, a matrix
	$\bvec{\Phi}\in \real^{M\times N}\ (N\ge K)$, is said to satisfy the restricted isometry property (RIP) of order $K$ if $\exists\delta\ge 0$, such that the following is satisfied for all $K-$sparse vectors $\bvec{x}\in \real^N$:
	\begin{align*}
	(1-\delta)\norm{\bvec{x}}^2\le \norm{\bvec{\Phi x}}^2\le (1+\delta)\norm{\bvec{x}}^2.
	\end{align*}
	The smallest such constant $\delta$ is denoted by $\delta_K$, and is called the restricted isometry constant
	(RIC)~\cite{candes_decoding_2005} of the matrix $\bvec{\Phi}$ of order $K$.
\end{define}
\subsection{An useful lemma}

\label{sec:preliminaries-lemmas}
%
%\begin{lem}[Theorem $2$ of~\cite{gribonval2007highly}]
%%
%\label{lem:NSP-implies-uniqueness}
%%
%If $\gamma(J,\bvec{\Phi},K)<1$, then any $\bvec{x}$ satisfying $\opnorm{\bvec{x}}{0}\le K$ and $\bvec{y}=\bvec{\Phi x}$, is the unique solution of~\eqref{eq:J-minimization-problem}.
%\end{lem}
%
%We further need the following lemma concerning the verification of a vertex of a polyhedron, which will be later useful in the subsequent analysis:
\begin{lem}[\cite{bertsekas2009convex}, pp.95, Prop. 2.1.4(a)]
	\label{lem:verification-vertex-polyhedron}
	If $P$ is a polyhedral subset of $\real^n$ of the form \begin{align}
	P=\{\bvec{x}:\bvec{a}_j^t\bvec{x}\le b_j,\ j=1,2,\cdots,r\},
	\end{align}
	where $\bvec{a}_j\in\real^n,\ b_j\in \real,\ j=1,\cdots,\ r$, then a vector $\bvec{v}\in P$ is an extreme point of $P$ if and only if the set \begin{align}
	A_v=\{\bvec{a}_j:\bvec{a}_j^t\bvec{v}=b_j,\ j\in \{1,\cdots,r\}\}
	\end{align}
	contains $n$ linearly independent vectors.
\end{lem}
\subsection{Elasticity of a sparsity promoting function}
\label{sec:elastcity-sparsity-promoting-function}
\begin{define}
	\label{def:elasticity-function}
	The \emph{elasticity}~\cite{carvajal2016essential} of a sparsity promoting function $f$ is a non-negative real valued function $\psi:\real_+\to \real_+$, defined as $\psi(x)=xf'(x)/f(x),\ x> 0$. The \emph{maximum elasticity} and \emph{minimum elasticity} are defined as $L_{\psi}=\sup_{x>0}\psi(x)$ and $l_\psi=\inf_{x>0}\psi(x)$, respectively.   
\end{define}
%The notion of elasticity is frequently encountered in the literature of mathematical economics where it is used to measure the ratio of fractional changes in demand to fractional changes in price~\cite{knut2016essential}. In the context of sparsity promoting functions, the elasticity $\psi$ is a measure of the susceptibility of sparsity to change in the magnitude of a coordinate. A sparsity promoting function might be highly elastic for small changes in coordinate but very inelastic for large changes in coordinates. This function is going to be very useful later in the paper.     
% 
%
\subsection{The scale function}
\label{sec:scale-function}
In this section, we introduce the \emph{scale} function associated with a sparsity promoting function. This function plays a critical role later in the analysis.
\begin{define}
	\label{def:scale-function-definition}
	The \emph{scale} function $g:\real_+\to \real_+$  of $f$ is defined as $g(x)=\sup_{y>0}\frac{f(xy)}{f(y)}$.
\end{define}
\begin{lem}[Properties of scale function]
	\label{lem:scale-function-properties}
	For any function $f$ satisfying the properties in Definition~\ref{def:f-properties}, the corresponding scale function $g$ satisfies the following properties: 
	\begin{enumerate}
		\item $g(0) = 0$.
		\item $g$ is non-decreasing in $(0,+\infty)$.
		\item $g(x)/x$ is non-increasing in $(0,\infty)$.
		\item $g(xy)\le g(x)g(y),\ \forall x,y\ge 0$.
		\item $g(x)\in[x\wedge 1,x\vee 1]$.
	\end{enumerate}
\end{lem}
\begin{remark}
	The first three properties of Lemma~\ref{lem:scale-function-properties} makes the function $g$ similar to a \emph{a sparsity promoting function}, however it might be discontinuous at $0$; the fourth property shows that $g$ is sub-multiplicative, and the last property shows that $g(x)$ is finite for each $x>0$.
\end{remark}
\begin{proof}
	The part 1) follows simply from $g(0) = \sup_{y>0}\frac{f(0\cdot y)}{f(y)}=0$.\\
	For part 2), choose $0\le x_1\le x_2$, and any $y>0$. Using monotonicity property of $f$, $\frac{f(x_1 y)}{f(y)} \le \frac{f(x_2 y)}{f(y)}\implies g(x_1)\le g(x_2).$\\
	For part 3), note that for any $0<x_1\le x_2$, and for any $y>0$, it follows from property $(3)$ of $f$ that $\frac{f(x_1y)}{x_1f(y)}=\frac{y}{f(y)}\frac{f(x_1y)}{x_1y} \ge \frac{y}{f(y)}\frac{f(x_2y)}{x_2y}=\frac{f(x_2y)}{x_2f(y)}
	\implies \frac{g(x_1)}{x_1} \ge \frac{g(x_2)}{x_2}.$\\ 
	For part 4), first assume that one of $x,y$ is $0$, so that $g(xy)=0=g(x)g(y)$. Then, with $x,y>0$, $g(xy)=\sup_{z>0}\frac{f(xyz)}{f(z)}  =\sup_{z>0}\frac{f(xyz)}{f(xz)}\frac{f(xz)}{f(z)}
	\le  \sup_{u>0}\frac{f(uy)}{f(u)}\sup_{z>0}\frac{f(xz)}{f(z)} =g(y)g(x).$\\
	Finally, for part 5), we consider the cases $x\le 1$, and $x>1$. If $x\le 1$, $f(xy)\le f(y)$, and $\frac{f(xy)}{xy}\ge \frac{f(y)}{y} \implies f(xy)\ge xf(y).$ Hence, for $x\le 1, g(x)\in[x,1]$. Using similar reasoning, one obtains $g(x)\in[1,x]$ when $x>1$.
\end{proof}
\subsection{Computation of the scale function}
%\section{Computing the scale function $g$}
%
\label{sec:computing-scale-function}
In order to obtain an expression of $g(y)$, we need to study the function $\phi_y(x)=f(xy)/f(x),\ \forall x>0$, for a fixed $y>0$. Since $\phi_1(x)=1,\ \forall x$, we immediately have $g(1)=1$. Now, to compute $g$ for $y\ne 1$, we find \begin{align}
\label{eq:derivative-phi-y}
\phi_y'(x)=\frac{f(xy)(\psi(xy)-\psi(x))}{xf(x)}.
\end{align}  
Clearly, if (for a fixed $y$) $\psi$ is a monotonic function,  the maximum value of $\phi_y(x)$ is obtained either at $x=0$ or at $x=\infty$. For example, if $\psi$ is a monotonically increasing function, $\phi_y'(x)\le 0$ for $y< 1$, and $\phi_y'(x)\ge 0$, for $y>1$, so that $g(y)=\lim_{x\downarrow 0}f(xy)/f(x)$ when $y< 1$, and $g(y)=\lim_{x\uparrow \infty}f(xy)/f(x)$. Similarly, when $\psi$ is a monotonically decreasing function, one can show that $g(y)=\lim_{x\uparrow \infty}f(xy)/f(x)$, when $y<1$, and $g(y)=\lim_{x\downarrow 0}f(xy)/f(x)$, when $y>1$.

When $\psi$ is not monotone, the function $\phi_y$ might possess several (probably even countably infinite) local maxima, making the computation of $g(y)$ difficult. However, in some special cases, efficient computation might still be possible. 

Let us assume that the set $S_\psi=\{x\ge 0:\psi'(x)=0\}$ of stationary points of $\psi$ can be efficiently computed. Let us define the supremum of $S_\psi$ as $x_\psi:=\sup S_\psi$. Now consider the case $x_\psi<B<\infty$, for some $B>0$. Note that for any local maxima $x_l$ ($\ne 0,\infty$) of $\phi_y$ (for a fixed $y$), $\psi(x_l y)=\psi(x_l)$, which implies by Rolle's Theorem that $\exists c\in[x_l\wedge x_ly,x_l\vee x_l y]$, such that $\psi'(c)=0$. Since $c\le x_\psi<B$, it follows that $x_l\in[0,B/y]$ if $y<1$, and $x_l\in [0,B]$ if $y>1$. Therefore, if $x_\psi<\infty$, all the local maxima ($\ne 0,\infty$) of $\phi_y(x)$ (for fixed $y$) belong to a compact set. Now, since $f(x)$ is concave and positive valued, $1/f(x)$ is a convex function in $(0,\infty)$. Therefore, maximization of $\phi_y$, when $x_{\psi}$ is bounded, is equivalent to finding the maximum of the ratio of two convex function in a compact set along with the points $0$ and $\infty$. The branch and bound algorithm in~\cite{benson2006maximizing} can then be used to find the global maxima of $\phi_y$ and hence the value of $g(y)$.

On the other hand, if $x_{\psi}=\infty$, the algorithm from~\cite{benson2006maximizing} cannot be applied. In that case one can proceed to find an upper and a lower bound on $g(y)$ in the following way. Assume that there is a procedure available to efficiently compute $L_\psi$ and $l_\psi$ (see Definition~\ref{def:elasticity-function}). It is easy to observe that $0\le l_\psi\le L_\psi\le 1$, so that $l_\psi$ and $L_\psi$ exist. Note that, due to concavity of $f$, for any $x>0$, $f(xy)\le f(x)+x(y-1)f'(x)$, so that
\begin{align}
\label{eq:scale-function-upper-bound-using-psi}
g(y)=\sup_{x>0}\frac{f(xy)}{f(x)} & \le 
\left\{\begin{array}{ll}
1+(y-1)L_\psi, & y> 1\\
1+(y-1)l_\psi. & y<1
\end{array}\right.
\end{align}
On the other hand, using $f(x)\le f(xy)+x(1-y)f'(xy)$ one obtains, \begin{align}
g(y) =\frac{1}{\inf_{x>0} \frac{f(x)}{f(xy)}} & \ge \frac{1}{1+\inf_{x>0}(1-y)\frac{xf'(xy)}{f(xy)}}\nonumber\\
\label{eq:scale-function-lower-bound-using-psi}
\implies g(y) & \ge \left\{\begin{array}{ll}
\frac{1}{1+\frac{(1-y)l_\psi}{y}}, & y>1\\
\frac{1}{1+\frac{(1-y)L_\psi}{y}}. & y<1
\end{array}\right.
\end{align}
\subsection{Examples of scale functions and elasticity functions}
\label{sec:example-computing-scale-elasticity-functions}
\begin{figure}[t!]
	\centering
	\begin{subfigure}{.5\textwidth}
		\centering
		\includegraphics[height=1.5in,width=3in]{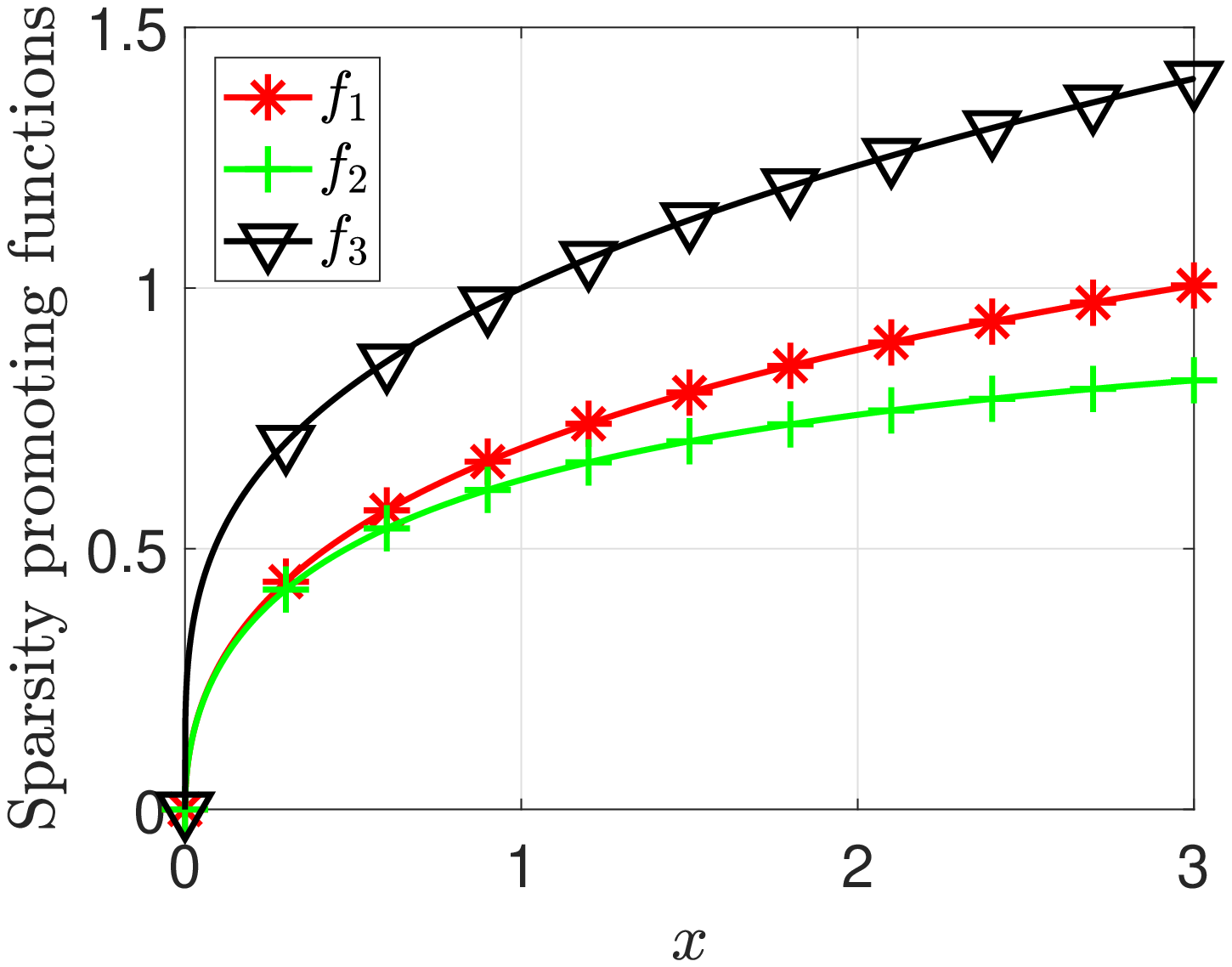}
		\caption{Sparsity promoting functions}
		\label{fig:sparsity-promoting-functions-example}
	\end{subfigure}
	\begin{subfigure}{.5\textwidth}
		\centering
		\includegraphics[height=1.5in,width=3in]{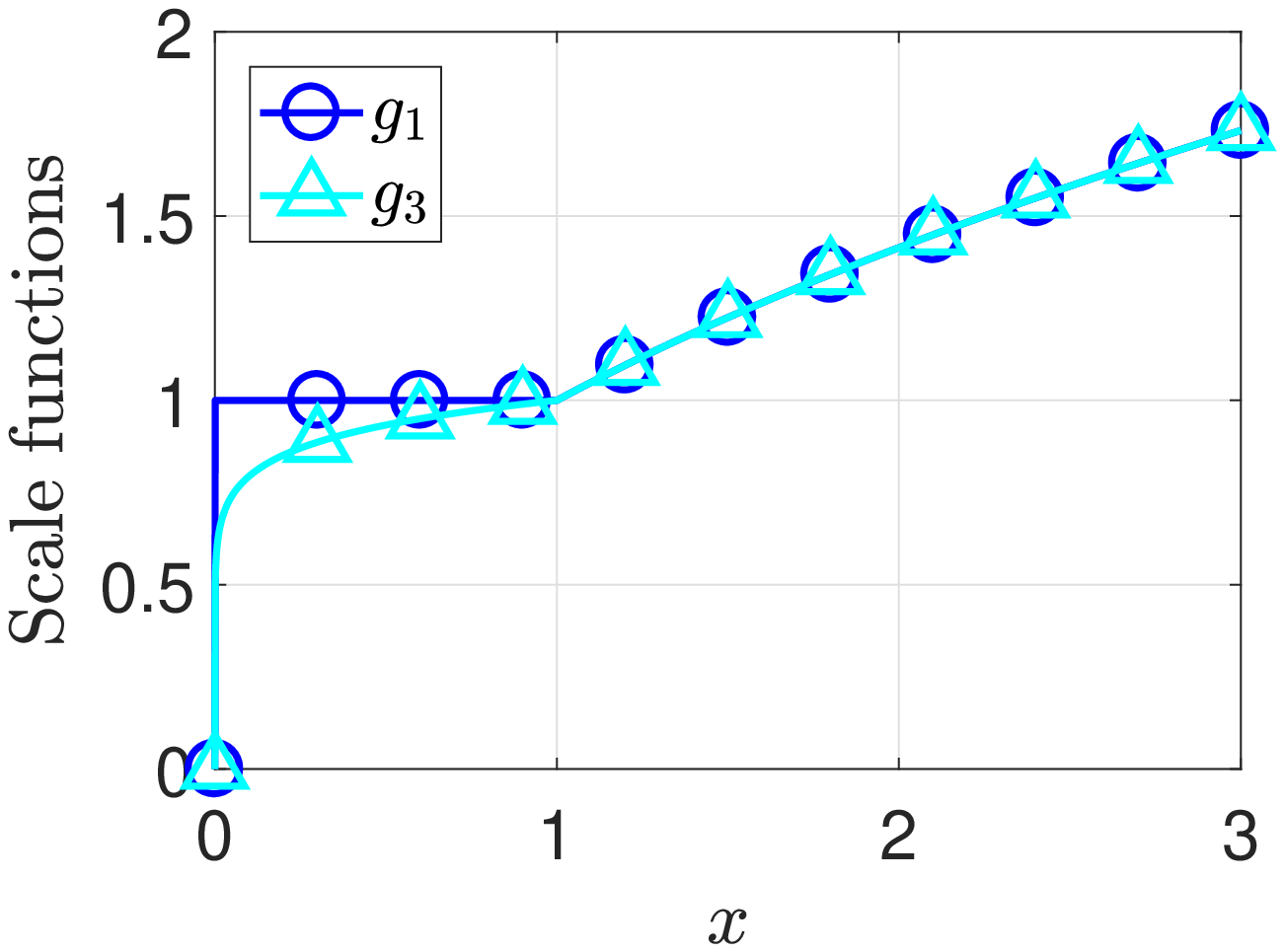}
		\caption{Scale functions}
		\label{fig:scale-functions}
	\end{subfigure}
	\begin{subfigure}{0.5\textwidth}
		\centering
		\includegraphics[height=1.5in, width=3in]{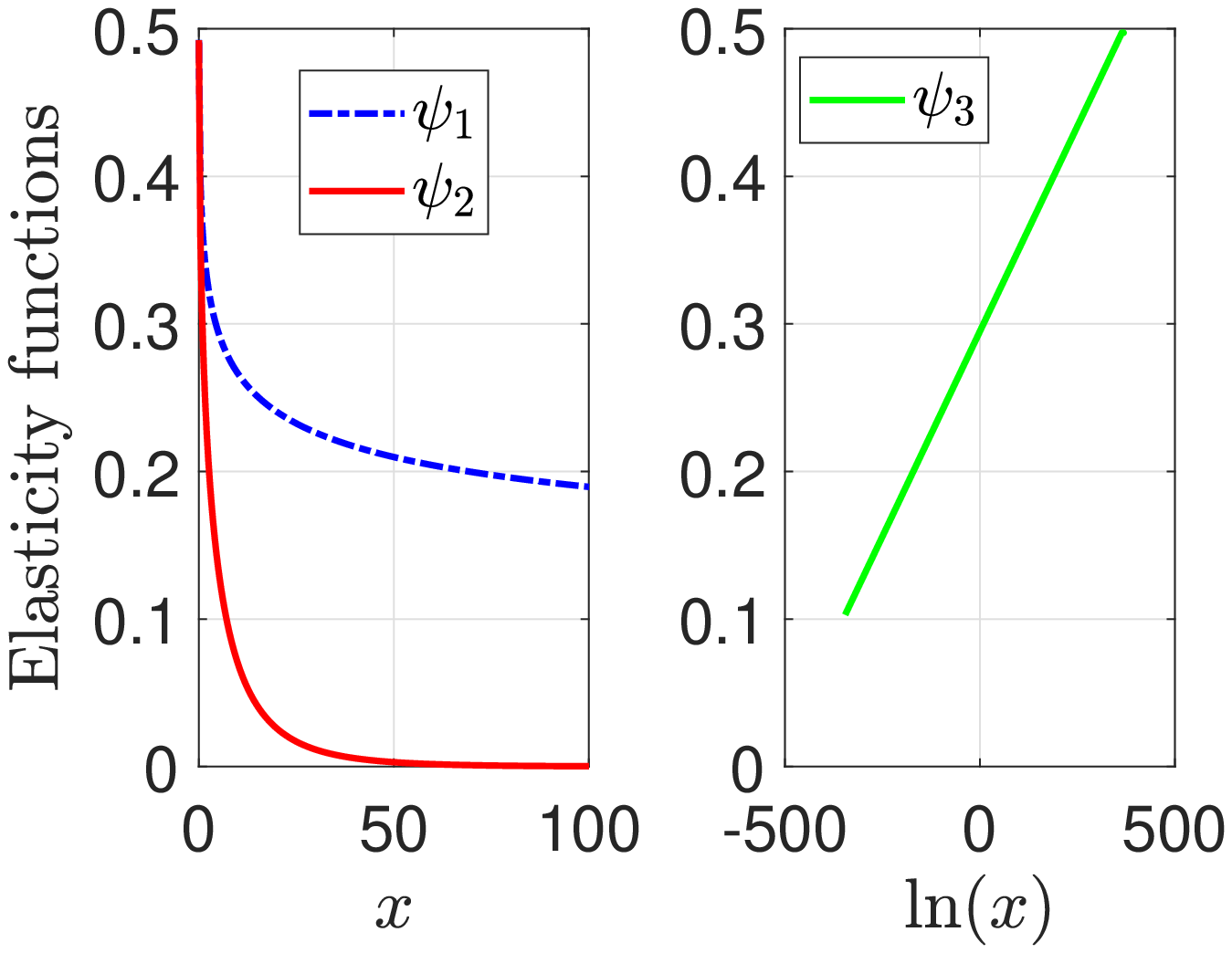}
		\caption{Elasticity functions}
		\label{fig:elasticity-functions}
	\end{subfigure}
	\caption{Examples of sparsity promoting functions (one-sided): $f_1(x)=\ln(1+\sqrt{x}),\ f_2(x)=1 - e^{-\sqrt{x}},\ f_3(x) = \int_{0.1}^{0.5}\frac{1}{0.4}x^p dp$, corresponding scale functions $g_1$ (associated with $f_1,f_2$) and $g_2$(associated with $f_3$) and the associated elasticity functions $\psi_1,\psi_2,\psi_3$ (corresponding to $f_1,f_2,f_3$ respectively) from Table~\ref{tab:scale-functions-examples}.}
	\label{fig:sparsity-promoting-functions-and-scale-functions-examples}
\end{figure} 
\begin{table}[t!]
	\centering
	\scriptsize
	\captionsetup{justification=centering}
	\caption{\textsc{Different Sparsity Promoting Functions with their Scale Functions}}
	\label{tab:scale-functions-examples}
	\begin{tabular}{|>{\centering\arraybackslash}m{0.1cm}|>{\centering\arraybackslash}m{1.5cm}|>{\centering\arraybackslash}m{2.5cm}|>{\centering\arraybackslash}m{1.75cm}|>{\centering\arraybackslash}m{0.25cm}|>{\centering\arraybackslash}m{0.25cm}|}
		\hline
		No. & $f(x)$ & $g(y)$ & $\psi(x)$ & $L_\psi$ & $l_\psi$\\
		%		\hline
		%		$1.$ & $\abs{x}^p$ & $y^p$ & $p$ & $p$\\
		\hline
		$1.$ & $\ln(1+\abs{x}^p)$ & $\indicator{y\le 1}+y^p\indicator{y> 1}$ & $\frac{px^p}{\ln(1+x^p)(1+x^p)}$ & $p$ & $0$\\
		\hline 
		$2.$ & $1-e^{-\abs{x}^p}$ & $\indicator{y\le 1}+y^p\indicator{y> 1}$ & $\frac{px^p}{e^{x^p}-1}$ & $p$ & $0$\\
		\hline
		$3.$ & $\expect{p}{\abs{x}^p}$, support $\Omega$, $p_1=\inf \Omega,\ p_2=\sup \Omega$ & $y^{p_1}\indicator{y\le 1}+y^{p_2}\indicator{y> 1}$ & $\frac{\expect{p}{px^p}}{\expect{p}{x^p}}$ & $p_2$ & $p_1$\\
		\hline
		%		$4.$ & $(2\sigma \abs{x}-\sigma^2x^2)\indicator{\abs{x}\le\frac{1}{\sigma}}+\indicator{\abs{x}>\frac{1}{\sigma}},(\sigma>0)$ & $\indicator{y\le 1}+y\indicator{y> 1}$\\
		%		\hline
	\end{tabular}
\end{table}
We will now evaluate the elasticity and scale functions for several important and frequently used sparsity promoting functions $f$. 

1)\textbf{$f(x)=\abs{x}^p,\ p\in [0,1]$:} This is the celebrated $l_p$ function which is a widely used and well-analyzed sparsity promoting function in the literature~\cite{gribonval2007highly, chartrand2008restricted,foucart2009sparsest,chen2014convergence,chen2015null,woodworth2016compressed,gu2018sparse-lp-bound}. Clearly, for this function $\psi(x)=p,\forall x>0$, which implies that $L_\psi = l_{\psi} = p$ and $g(y)=y^p,\ \forall y>0$.

2)\textbf{$f(x)=\ln(1+\abs{x}^p) (p\in (0,1])$:} This function is commonly known as the Lorentzian function in the literature~\cite{carrillo2010robust,carrillo2013lorentzian}. For this function, we find that, for all $x>0$, \begin{align}
 \psi(x) & =\frac{p x^p}{(1+ x^p)\ln(1+ x^p)}\\
 \implies \psi'(x) & =\frac{p^2x^{p-1}(\ln(1+ x^p)- x^p)}{(1+ x^p)^2(\ln(1+ x^p))^2}\le 0,
\end{align} since $1+x\le e^x,\ \forall x\ge 0$. Hence, for $y\in(0,1]$, $g(y)=\lim_{x\uparrow\infty}\ln(1+ x^py^p)/\ln(1+ x^p)=1$, and for $y>1$, $g(y)=\lim_{x\downarrow 0}\ln(1+ x^py^p)/\ln(1+ x^p)=y^p$. Furthermore, $L_\psi=\lim_{x\downarrow 0}\psi(x)=p,\ l_\psi=\lim_{x\uparrow\infty}\psi(x)=0$.

3)\textbf{$f(x)=1-e^{-\abs{ x}^p} (p\in(0,1])$:} This is the so-called concave exponential function~\cite{bradley1998feature,fung2002minimal,tipping-etal2003l0}. In this case, $\forall x>0,\ \psi(x)= px^pe^{- x^p}/(1-e^{- x^p})=p/(\sum_{k=0}^{\infty}( x^p)^{k}/(k+1)!)$, which is obviously a monotonically decreasing function. Hence, in this case, for $y\in(0,1],\ g(y)=\lim_{x\uparrow \infty}(1-e^{- x^py^p})/(1-e^{-x^p})=1$, and for $y>1$, $g(y)=\lim_{x\downarrow 0}(1-e^{- x^py^p})/(1-e^{- x^p})=y^p$. Furthermore, $L_\psi=\lim_{x\downarrow 0}\psi(x)=p,\ l_\psi=\lim_{x\uparrow\infty}\psi(x)=0$.

4)\textbf{$f(x)=\expect{p}{\abs{x}^p}$:} This is the continuous mixed norm function introduced in~\cite{zayyani2014continuous} for $p\in[1,2]$ and later extended in~\cite{javaheri2018robust} to $p\in[0,2]$ (where it becomes a quasi-norm for $p\in (0,1)$). Here we use the function $f(x)=\expect{p}{\abs{x}^p}=\int_0^1\abs{x}^pd\nu$ for some probability measure $\nu$ defined on the Borel $\sigma$-field associated with the set $[0,1]$. In this case $\forall x>0$, $\psi(x)=\frac{\expect{p}{px^{p}}}{\expect{p}{x^p}}$, so that \begin{align}
\psi'(x)=\frac{\expect{p}{p^2x^{p}}\expect{p}{x^p}-(\expect{p}{px^p})^2}{x(\expect{p}{x^p})^2}\ge 0,
\end{align} since $\expect{p}{p^2x^{p}}\expect{p}{x^p}\ge (\expect{p}{px^p})^2$ by the Cauchy-Schwartz inequality. Hence, for $y\in (0,1]$, $g(y)=\lim_{x\downarrow 0}\frac{\expect{p}{x^py^p}}{\expect{p}{x^p}}$, and for $y>1$, $g(y)=\lim_{x\uparrow \infty}\frac{\expect{p}{x^py^p}}{\expect{p}{x^p}}$. Furthermore, $l_\psi = \lim_{x\uparrow \infty }\psi(x),\ l_\psi=\lim_{x\downarrow 0}\psi(x)$. It can be shown that if $\Omega(\subseteq [0,1])$ is the support of the measure $\nu$, and $p_1=\inf\Omega,\ p_2=\sup \Omega$, with $0<p_1\le p_2<1$, 
\begin{align}
	\label{eq:g(y)-evluation-continuous-mixed-norm}	
	\lim_{x\downarrow 0}\frac{\expect{p}{x^py^p}}{\expect{p}{x^p}} & = y^{p_1},\	\lim_{x\uparrow \infty}\frac{\expect{p}{x^py^p}}{\expect{p}{x^p}} = y^{p_2},\\
	\label{eq:psi(y)-evluation-continuous-mixed-norm}
	\lim_{x\downarrow 0} \frac{\expect{p}{px^p}}{\expect{p}{x^p}} & = p_1,\ 	\lim_{x\uparrow\infty} \frac{\expect{p}{px^p}}{\expect{p}{x^p}}=p_2.
\end{align} See Appendix~\ref{sec:appendix-proving-equations-g(y)-continuos-mixed-norm} for a proof of these claims.

The above sparsity promoting functions as well as the associated elasticity functions and scale functions are tabulated in Table.~\ref{tab:scale-functions-examples} and examples are illustrated in Fig.~\ref{fig:sparsity-promoting-functions-and-scale-functions-examples}.

%5)\textbf{$f(x)=(2\sigma\abs{x}-\sigma^2x^2)\indicator{\abs{x}\le \frac{1}{\sigma}}+\indicator{\abs{x}> \frac{1}{\sigma}}:$} This is the minimax concave penalty (MCP) which was proposed in~\cite{zhang2010nearly} for regularized variable selection for linear regression and has been extensively used in sparse logistic regression~\cite{loh2013regularized} and compressed sensing~\cite{shen2016square}. For this function, we have, for $x>0$, $
%\psi(x) =2\left(1-\frac{1}{2-\sigma x}\right)\indicator{x\le \frac{1}{\sigma}}$, so that $\psi$ is monotonically decreasing for $0<x\le 1/\sigma$, and $\psi(x)=0,\ \forall x>1/\sigma$. Let $y\in (0,1]$. Then, we have to consider three cases to evaluate $g(y)$: $x\le \frac{1}{\sigma},\ x\in(\frac{1}{\sigma},\frac{1}{\sigma y}]$ and $x>\frac{1}{\sigma y}$. Clearly, $\sup_{x\le \frac{1}{\sigma}}\phi_y(x)=\frac{f(\frac{y}{\sigma})}{f(\frac{1}{\sigma})}=f(\frac{y}{\sigma})=2y-y^2\le 1$; $\sup_{x\in(\frac{1}{\sigma},\frac{1}{\sigma y}]}\phi_y(x)=\sup_{x\in(\frac{1}{\sigma},\frac{1}{\sigma y}]}\frac{f(xy)}{1}=f(\frac{1}{\sigma})=1 $; $\sup_{\frac{1}{y\sigma}<x}\phi_y(x)=1$. Hence, $g(y)=1$, if $y\le 1$. Using similar manipulations it can be shown that $g(y)=y$, if $y>1$.
%
%
\section{Upper bound on $J$-NSC}
\label{sec:upper-bound-F-NSP}
Before stating the main results of the paper, we state a lemma which lies at the heart of the main result of this paper.
\begin{lem}
	\label{lem:generalized-holder-type-bounds}
	Let $f$ be a function satisfying the properties~\ref{def:f-properties}. Let $u_0=\inf\{u\in(0,1]:g(u)=1\}$ and for any positive integer $s\ge 1$, and real number $q\ge 1$, let the non-negative real numbers $\alpha_f(q,s),\ \beta_f(q,s)$ satisfy the following: \begin{enumerate}
		\item $\alpha_f(q,s)+\beta_f(q,s)\ge 1$,
		\item If $u_0=0$, $\beta_f(q,s)\ge (1-1/s)^{1/q}$,
		\item If $u_0>0$, $\displaystyle\sup_{u\in(0,u_0)}\frac{g(u)}{(\alpha_f(q,s) u+\beta_f(q,s))^q}\le 1.$
%		where the latter is used to choose $\beta_f(q,s)$ as the smallest non-negative number satisfying the above, if it can yield a value smaller than $(1-1/s)^{1/q}$.  
	\end{enumerate}  
	Then, for any vector $\bvec{a}\in \real^s$, the following is satisfied: \begin{align}
	\label{eq:main-holder-type-inequality-generalized-concave-function}
	f^{-1}\left(\frac{\sum_{j=1}^s f(a_j)}{s}\right) \le \frac{\opnorm{\bvec{a}}{q}}{s^{1/q}} & \le \alpha_f(q,s) f^{-1}\left(\frac{\sum_{j=1}^s f(a_j)}{s}\right) +\beta_f(q,s)\opnorm{\bvec{a}}{\infty}.
	\end{align}
\end{lem}
\begin{proof}
	See Appendix~\ref{sec:appendix-proof-of-lemma-generalized-holder-type-bounds}.
\end{proof}
The Lemma~\ref{lem:generalized-holder-type-bounds} demonstrates that for a function $f$ satisfying the conditions~\ref{def:f-properties} of a sparsity promoting function, there is positive linear combination of the $l_\infty$ norm of a non-negative vector $\bvec{a}$, and its generalized $f$-mean, such that the combination is no smaller than its $l_q$ mean for $q\ge 1$. This result generalizes Lemma 1 in~\cite{gu2018sparse-lp-bound}, that uses $f(x)=\abs{x}^p,\ p\in(0,1]$. To verify this, note that when $f(x)=\abs{x}^p,\ p\in (0,1]$, one obtains $g(x)=x^p$ so that $\beta_f$ in Lemma~\ref{lem:generalized-holder-type-bounds} can be obtained by solving the following: \begin{align}
\label{eq:alpha-beta-evaluation-for-lp-funtion}
\sup_{u\in [0,1]} \frac{u^p}{(\alpha_f u + \beta_f)^q} \le 1 &
\Leftrightarrow \inf_{u\in [0,1]} \left(\alpha_f u + \beta_f - u^{p/q}\right) \ge 0.
\end{align}
Since the function $u^{p/q}$ is concave for $u>0$, the function in~\eqref{eq:alpha-beta-evaluation-for-lp-funtion} is convex and the minimum is obtained when $u=\left(\frac{p}{q\alpha}\right)^{1/(1-p/q)}$, so that \begin{align}
\inf_{u\in [0,1]} \left(\alpha_f u + \beta_f - u^{p/q}\right) \ge 0
\implies & \beta_f^{1-r} \alpha_f ^r \ge r^r(1-r)^{1-r},
\end{align} which implies that $c^rd^{1-r}\ge 1,$ where $r=p/q,\ c=\frac{\alpha_f}{r},\ d = \frac{\beta_f}{1-r}$. This is exactly the condition stated in Eq(12) of~\cite{gu2018sparse-lp-bound} with the equality sign. Note that for such a pair of $\alpha_f,\beta_f$, one has $\alpha_f+\beta_f\ge 1$, since $\alpha_f+\beta_f=cr+d(1-r)=\exp(\ln(cr+d(1-r)))\ge c^rd^{1-r}\ge 1$, where the penultimate step uses Jensen's inequality. The Lemma~\ref{lem:generalized-holder-type-bounds} enables us to find upper bounds on the null space constant for generalized sparsity promoting functions beyond $l_p$, which we believe will be useful in analyzing the performances of sparse recovery problems with general sparsity promoting functions. Using Lemma~\ref{lem:generalized-holder-type-bounds}, we now present the main result of this paper in the form of the following theorem:
\begin{thm}
	\label{thm:null-space-constant-upper-bound}
	For any given integer $K_0\ge 1$, and any matrix $\bvec{\Phi}$ satisfying $\delta_{2K_0}<1$, for any positive integer $1\le K\le K_0$ and $p\in(0,1]$, the following holds for any sparsity promoting function $f$:
	\begin{align}
	\label{eq:null-space-constant-upper-bound}
		\gamma(J,\bvec{\Phi},K)< \gamma^*(J,\bvec{\Phi},K),
	\end{align}
	where, $J(\bvec{x})=\sum_{i=1}^N f(x_i)$ for any $\bvec{x}\in \real^N$, $L'=2K_0-K+1$, and
	 \begin{align}
	\label{eq:null-space-constant-upper-bound-expression}
	\gamma^*(J,\bvec{\Phi},K) & =\frac{K}{L'}g(\zeta_f(K,K_0)),\\
	\label{eq:zeta-f-expression}
	\zeta_f(K,K_0) & =\frac{L'(\pi_f(L')+\beta_f(1,L'))C_{2K_0}'}{\sqrt{K}\sqrt{L'-1}},\\
	\label{eq:C-2K0-expression}
	C_{2K_0}' & =\frac{\sqrt{2}+1}{2}\frac{\delta_{2K_0}}{1-\delta_{2K_0}},\\
	\label{eq:pi-f-expression}
	\pi_f(L') & =\max\left\{\alpha_f(1,L'),\frac{1}{2}+\frac{1}{2L'}\right\}.
	\end{align}
\end{thm}
\begin{proof}
	See Appendix~\ref{sec:appendix-proof-of-theorem-null-space-constant-upper-bound}.
\end{proof}
Using the expressions for the scale function $g$, as tabulated in Table~\ref{tab:scale-functions-examples}, the result of Theorem~\ref{thm:null-space-constant-upper-bound} can be specialized for the different sparsity promoting functions described in Table~\ref{tab:scale-functions-examples}, in the form of the following corollary:
\begin{cor}
		\label{cor:null-space-constant-upper-bound-different-functions}
		For any given integer $K_0\ge 1$, and any matrix $\bvec{\Phi}$ satisfying $\delta_{2K_0}<1$, and for any integer $1\le K\le K_0$:
		\begin{enumerate}
			\item If $f(x)=\ln(1+\abs{x}^p)$, or $f(x)=1-e^{-\abs{x}^p}$, the following holds, 
			 \begin{align}
			 \label{eq:null-space-constant-upper-bound-expression-exp-or-log}
			\gamma^*(J,\bvec{\Phi},K) & = \frac{K}{L'}\left[\indicator{\zeta_1\le 1}+ \zeta^p_1\indicator{\zeta_1> 1}\right],
			\end{align}
			where $\zeta_1=\frac{3L'-1}{2\sqrt{K}\sqrt{L'-1}}C'_{2K_0}$, $L'=2K_0-K+1$.
%			\item If $f(x)=(2\sigma\abs{x}-\sigma^2x^2)\indicator{\abs{x}\le 1}+)\indicator{\abs{x}> 1}$ the following holds, 
%			\begin{align}
%				\label{eq:null-space-constant-upper-bound-expression-mcp} 
%			 \gamma(J,\bvec{\Phi},K) & \le \frac{K}{L'}\left[\indicator{\zeta_1\le 1}+\zeta_1\indicator{\zeta_1> 1}\right].
%			\end{align}
			\item If $f(x)=\expect{p}{\abs{x}^p}$, with $\Omega$ as the support of the probability measure and $p_1=\inf \Omega,\ p_2=\sup \Omega$, the following holds, 
			\begin{align}
			\label{eq:null-space-constant-upper-bound-prelim-expression-continuous-mixed-norm}
			 \gamma^*(J,\bvec{\Phi},K)  & = \frac{K}{L'}\left[\widehat{\zeta}_2^{p_1}\indicator{\widehat{\zeta}_2\le 1}+\widehat{\zeta}_2^{p_2}\indicator{\widehat{\zeta}_2> 1}\right]\\
			 \label{eq:null-space-constant-upper-bound-expression-continuous-mixed-norm}
			 \  & = \frac{K}{L'}\left[\zeta_2^{p_1}\indicator{\zeta_2\le 1}+\zeta_2^{p_2}\indicator{\zeta_2> 1}\right],
			\end{align}
			where $\widehat{\zeta}_2 = \frac{\max\{2L',3L'-2p_1L'+1\}}{2\sqrt{K}\sqrt{L'-1}}C'_{2K_0}<\zeta_2 = \frac{3L'+1}{2\sqrt{K}\sqrt{L'-1}}C_{2K_0}'$, $L'=2K_0-K+1$.
%			where $\zeta_2 = \frac{\max\{L',(3-2p_1)L'+1\}}{2\sqrt{K}\sqrt{L'-1}}C_{2K_0}'$, $L'=2K_0-K+1$.
		\end{enumerate}
\end{cor}
\begin{proof}
	
%	The proof follows by writing down the expression for $\gamma^\star(J,\bvec{\Phi},K)$ in the RHS of inequality~\eqref{eq:null-space-constant-upper-bound} by evaluating the quantity $\zeta_f$ in Eq~\eqref{eq:zeta-f-expression} for the different functions, and writing down the expression for the corresponding scale function as given in Table~\ref{tab:scale-functions-examples}. 
	To prove Eq.~\eqref{eq:null-space-constant-upper-bound-expression-exp-or-log}, note that  $g(u)=1,\ \forall\ u\in(0,1]$ so that $u_0=0$ for both the functions $1,2$ in Table~\ref{tab:scale-functions-examples}. Therefore, by Lemma~\ref{lem:generalized-holder-type-bounds}, $\beta_f(1,L')=1-1/L'$, and $\alpha_f(1,L')=1/L'$, implying that $\pi_f(L')=1/2+1/(2L')$, since $L'>1$. Consequently, one obtains, that $\zeta_f(K,K_0)=\zeta_1$ for the functions $1,2$.\\
	To prove the bound~\eqref{eq:null-space-constant-upper-bound-expression-continuous-mixed-norm}, note that for the function $3$ in Table~\ref{tab:scale-functions-examples}, $g(u)=u^{p_1},\ u\in [0,1]$, so that $u_0=1$. Therefore, one can use Lemma~\ref{lem:generalized-holder-type-bounds} to choose $\beta_f(1,L')=1-p_1,\ \alpha_f(1,L')=p_1$, so that , $\pi_f(L')+\beta_f(1,L')=\max\left\{1,3/2-p_1+1/(2L')\right\}<3/2+1/(2L')$. Consequently, in this case, $\zeta_f(K,K_0)=\widehat{\zeta}_2<\zeta_2$, and thus, $\gamma^\star(J,\bvec{\Phi},K)=K/L'g(\widehat{\zeta}_2)<K/L'g(\zeta_2)$ from Eq.~\eqref{eq:null-space-constant-upper-bound-expression}.
\end{proof}
We have following remarks to make on Corollary~\ref{cor:null-space-constant-upper-bound-different-functions}:

\textbf{1)} We find from Equations~\eqref{eq:null-space-constant-upper-bound-expression-exp-or-log} and~\eqref{eq:null-space-constant-upper-bound-prelim-expression-continuous-mixed-norm} that \begin{align}
\label{eq:null-space-constant-tight-upper-bound-p-to-0-exp-log}
\lim_{p\downarrow 0} \gamma^\star(J,\bvec{\Phi},K) & = \frac{K}{2K_0-K+1},
\end{align} 
for the functions $1$ and $2$ in Table~\ref{tab:scale-functions-examples}, and \begin{align}
\label{eq:null-space-constant-tight-upper-bound-p-to-0-continuous-mixed-norm}
\lim_{p_2\downarrow 0} \gamma^\star(J,\bvec{\Phi},K) & = \frac{K}{2K_0-K+1},
\end{align}
for the function $3$ in Table~\ref{tab:scale-functions-examples}. 
%Furthermore, from Theorem~5 in~\cite{gribonval2007highly}, and Theorem 2 in~\cite{chen2015null} it follows that, if $\mathrm{Spark}(\bvec{\Phi})=2K_0-K$\footnote[2]{$\mathrm{Spark}(\bvec{\Phi})$ is the smallest positive integer $r$ such that any collection of $r$ columns of $\bvec{\Phi}$  is linearly dependent.}, if 
Now, according to Theorem~5 in~\cite{gribonval2007highly}, if $J(\bvec{x})=\sum_{j=1}^N f(x_j), B(\bvec{x})=\sum_{j=1}^N b(x_i)$, and $f=a\circ b$, for two sparsity promoting functions $a,b$, then, \begin{align}
\label{eq:nsc-comparison-ine1quality}
\gamma(l_0,\bvec{\Phi},K)\le \gamma(J,\bvec{\Phi},K)\le \gamma(B,\bvec{\Phi},K),\ \forall K\ge 1.
\end{align}  Denoting the $J$ functions corresponding to the functions $1,2,3$ by $J_1,J_2,J_3$ and noting that for the functions $1,2,3$ in Table~\ref{tab:scale-functions-examples}, $f=a\circ l_p$, with $a(x)=\ln(1+\abs{x}),\ a(x)=1-e^{-\abs{x}}$, and $a(x)=\int \abs{x}^{p/p_2}d\nu,\ p\in[p_1,p_2]$, respectively, one obtains using the inequalities~\eqref{eq:nsc-comparison-ine1quality}, for all $K\ge 1$, \begin{align}
\ & \gamma(l_0,\bvec{\Phi},K)\le \gamma(J_i,\bvec{\Phi},K)\le \gamma(l_p,\bvec{\Phi},K),\ i=1,2\\
\ & \gamma(l_0,\bvec{\Phi},K)\le \gamma(J_3,\bvec{\Phi},K)\le \gamma(l_{p_2},\bvec{\Phi},K).
\end{align} Now, according to Theorem 2 of~\cite{chen2015null}, if $\mathrm{Spark}(\bvec{\Phi})=2K_0+1$\footnote[2]{$\mathrm{Spark}(\bvec{\Phi})$ is the smallest positive integer $r$ such that any collection of $r$ columns of $\bvec{\Phi}$ is linearly dependent.} for all $1\le K\le K_0$, \begin{align}
\lim_{p\downarrow }\gamma(l_p,\bvec{\Phi},K) & = \gamma(l_0,\bvec{\Phi},K) = \frac{K}{2K_0-K+1},
\end{align} where the last identity is taken from~\cite{foucart2013mathematical}. Therefore, $\forall 1\le K\le K_0$, $\lim_{p\downarrow 0}\gamma(J_i,\bvec{\Phi},K)(i=1,2)=\lim_{p_2\downarrow 0}\gamma(J_3,\bvec{\Phi},K)=\gamma(l_0,\bvec{\Phi},K)=\frac{K}{2K_0-K+1}$. Hence, as $p,p_2\to 0^+$, the upper bounds~\eqref{eq:null-space-constant-upper-bound-expression-exp-or-log}, and~\eqref{eq:null-space-constant-upper-bound-expression-continuous-mixed-norm} become tight, given that $\mathrm{Spark}(\bvec{\Phi})=2K_0+1$.

\textbf{2)} It is known from~\cite{gribonval2007highly} that exact recovery of any $K$-sparse vector can be done by solving the problem $P_f$ if and only if $\gamma(J,\bvec{\Phi},K)<1$. Thus, by requiring that $\gamma^*(J,\bvec{\Phi},K)<1$, one can obtain bounds on $\delta_{2K_0},$ given $p,K,K_0$; on $p,$ given $K,K_0,\delta_{2K_0}$; and on $K,$ given $K_0,p,\delta_{2K_0}$, under which exact recovery is possible by solving the problem $P_f$.

We further observe that the bound~\eqref{eq:null-space-constant-upper-bound-expression-continuous-mixed-norm} becomes the bound for the NSC of $l_p$ minimization, when $p_1=p_2=p$. However, this bound is a little different than the bound proposed by~\cite{gu2018sparse-lp-bound} in Eq~(8)(9)(10)(11), as in our bound the numerator of $\zeta_2$ consists of $3L'+1$, whereas, according to the bound in~\cite{gu2018sparse-lp-bound}, it should be $3L'$. This mismatch seems to stem from the fact that in the second last step to arrive Eq~(55) in~\cite{gu2018sparse-lp-bound}, the authors seem to have made a typo by writing $1-p+\frac{\frac{1}{2}+\frac{1}{2}(2K_0-K)}{2K_0-K+1}$, whereas it should have been $1-p+\frac{1+\frac{1}{2}(2K_0-K)}{2K_0-K+1}$, which follows from the definition of $\mu$ in Eq~(53) therein. This adjustment would make their bound identical to our bound~\eqref{eq:null-space-constant-upper-bound-expression-continuous-mixed-norm} for $p_1=p_2=p$.   
\section{Recovery conditions on RIC, $p,K$ for a few special functions}
\label{sec:recovery-conditions-ric-p-k}
In this section, we obtain bounds on $\delta_{2K_0},p,K$, for the three functions in Table~\ref{tab:scale-functions-examples} and discuss the implications of these bounds that ensure exact recovery by solving the problem $P_f$. 
\begin{thm}
	\label{thm:delta-upper-bound-perfect-recovery}
	For any $p\in[0,1]$, and any integer $1\le K\le K_0$, any $K-$sparse vector can be recovered exactly by solving $P_f$:
	\begin{enumerate}
		\item with the functions $1,2$ in Table~\ref{tab:scale-functions-examples} if, \begin{align}
		\label{eq:delta-upper-bound-perfect-recovery-exp-log}
		\delta_{2K_0} & < f_1(K,K_0,p):=\frac{1}{1+b(K,K_0)\left(\frac{K}{K_0+1}\right)^{1/p}},
		\end{align}
		where $b(K,K_0)=\left(3-\frac{1}{2K_0-K+1}\right)(\sqrt{2}+1)/4$;
		\item with the function $3$ in Table~\ref{tab:scale-functions-examples}, if \begin{align}
		\label{eq:delta-upper-bound-perfect-recovery-continuous-mixed-norm}
		\delta_{2K_0} & < f_2(K,K_0,p_2):=\frac{1}{1+d(K,K_0)\left(\frac{K}{K_0+1}\right)^{1/p_2}},
		\end{align}
		where $d(K,K_0)=\left(3+\frac{1}{2K_0-K+1}\right)(\sqrt{2}+1)/4$;
	\end{enumerate}
\end{thm}
\begin{proof}
	For the part 1), the proof follows by noting that if $\zeta_1\le 1$, then obviously, from~\eqref{eq:null-space-constant-upper-bound-expression-exp-or-log} $\gamma(J,\bvec{\Phi},K)<1$, for all $K\le K_0$, whereas, if $\zeta_1>1$, then $\gamma(J,\bvec{\Phi},K)<1$ if $\zeta_1<(L'/K)^{1/p}$. Since $L'>K,\ \forall K\le K_0$, it follows that a sufficient condition for ensuring $\gamma(J,\bvec{\Phi},K)<1$ is $\zeta_1<(L'/K)^{1/p}$. Finally, noting that $K\le K_0\implies K/(2K_0-K+1)\le K/(K_0+1)$, $\zeta_1<(L'/K)^{1/p}$ is ensured if~\eqref{eq:delta-upper-bound-perfect-recovery-exp-log} is satisfied. For part 2), using similar reasoning, one finds that a sufficient condition for showing $\gamma(J,\bvec{\Phi},K)<1$ is $\zeta_2<(L'/K)^{1/p_2}$, which is ensured if~\eqref{eq:delta-upper-bound-perfect-recovery-continuous-mixed-norm} is satisfied.  
\end{proof}
We make the following comments on the results of Theorem~\ref{thm:delta-upper-bound-perfect-recovery}:\\
%\begin{enumerate}
	\textbf{1)} Both the functions $f_1(K,K_0,p)$ and $f_2(K,K_0,p_2)$ are decreasing functions of $K$ for $1\le K\le K_0$. This is obvious for $f_2(K,K_0,p_2)$. For $f_1(K,K_0,p)$, it is straightforward to show that $\frac{d}{dK}(b(K,K_0)K^{1/p})=\frac{(\sqrt{2}+1)K^{1/p-1}}{4p}\left(3-\frac{2K_0-(1-p)K+1}{(2K_0-K+1)^2}\right)$. Since $(2K_0-K+1)^2-(2K_0-(1-p)K+1)=(2K_0-K+1)(2K_0-K)+pK>0$ whenever $0<K\le 2K_0$, it follows that $b(K,K_0)K^{1/p}$ is increasing in $K$, which implies that $f_1(K,K_0,p)$ is a decreasing function of $K$ for $1\le K\le K_0$. Therefore, as the sparsity becomes smaller, the RIC bound becomes larger and the recovery becomes easier.\\
	\textbf{2)} For $1\le K\le K_0$, the functions $f_1(K,K_0,p)$ and $f_2(K,K_0,p_2)$ are decreasing functions of $p,p_2$, so that, for smaller $p,p_2$, the RIC bound is larger and the recovery becomes easier. Specifically, if $p,p_2\to 0+$, the recovery conditions reduce to $\delta_{2K_0}<1$. Since, as $p,p_2\to 0+$ the functions $1,2$ in Table~\ref{tab:scale-functions-examples} reduce to a constant multiple of $l_0$ (with constants $\ln 2$, and $1-e^{-1}$, respectively) and the function $3$ in Table~\ref{tab:scale-functions-examples} reduces to $l_0$, the problems $P_f$ reduce to the problem $P_0$, for which the condition $\delta_{2K_0}<1$ coincides with the optimal unique sparse recovery condition in~\cite{foucart2013mathematical}.\\
%	\textbf{3)} For $K=K_0$, the conditions $\delta_{2K_0}<f_1(K_0,K_0,p)$ and $\delta_{2K_0}<f_2(K_0,K_0,p_2)$ ensure exact recovery by the problems $P_f$ with functions $1,2$ and $3$ respectively, when the sparsity is $K_0$. Since the functions $f_1(K_0,K_0,p),f_2(K_0,K_0,p_2)$ can be shown to be decreasing in $K_0$ as well as $p,p_2$, smaller sparsity and smaller $p,p_2$ makes the satisfaction of the recovery conditions easier for the recovery.\\
	\textbf{3)} Another interesting observation is that for given $K,K_0,p_2=p$, the RIC bound for the functions $1,2$ are larger than the RIC bound for function $3$. Therefore, exact recovery by solving $P_f$ is easier by using functions $1,2$ than by using function $3$. 
%	is easier tThis indicates that there are matrices for which, given $K,K_0,p_2=p$, $\delta_{2K_0}$ might be large enough for the problem $P_f$ with $f$ as the function $3$ to fail in exactly recovering the target $K-$sparse vector while, the RIC might be small enough for the problems $P_f$ with $f$ as the functions $1,2$ to be successful in exact recovery. This scenario will be further justified numerically in Section~\ref{sec:numerical-results}.
%\end{enumerate}

It is worth mentioning that both the bounds~\eqref{eq:delta-upper-bound-perfect-recovery-exp-log} and~\eqref{eq:delta-upper-bound-perfect-recovery-continuous-mixed-norm} are achievable by choosing $\bvec{\Phi}$ as a random matrix with  i.i.d. sub-Gaussian entries. In particular, for any $\epsilon\in (0,1)$, if one chooses $\bvec{\Phi}$ as a random matrix with i.i.d. Gaussian entries, with number of rows $M\ge 80.098(f_i(K,K_0,p))^{-2}(K_0\ln(Ne/K_0)+\ln(2/\epsilon))$, using Theorem 9.27 in~\cite{foucart2013mathematical}, the matrix $\bvec{\Phi}$ satisfies the bound~\eqref{eq:delta-upper-bound-perfect-recovery-exp-log} (for $i=1$) or the bound~\eqref{eq:delta-upper-bound-perfect-recovery-continuous-mixed-norm} (for $i=2$) with probability at least $1-\epsilon$. When $K,K_0$ are large, one can see from Theorem~\ref{thm:delta-upper-bound-perfect-recovery} that the quantities $b(K,K_0)$ and $d(K,K_0)$ become almost equal to the constant $\frac{3}{4}(\sqrt{2}+1):=D$, so that $f_i(K,K_0,p)^{-1}\approx 1+D(K/K_0)^{1/p}$. Furthermore, for fixed ratio $N/K_0$, $K_0\ln(Ne/K_0)=\mathcal{O}(K_0)$, so that, for large $K_0$, the number of rows of the Gaussian measurement matrix, required for satisfaction of the bounds~\eqref{eq:delta-upper-bound-perfect-recovery-exp-log} and~\eqref{eq:delta-upper-bound-perfect-recovery-continuous-mixed-norm} with high probability, scale as $(1+D(K/K_0)^{1/p})^2\mathcal{O}(K_0)$. 
\begin{thm}
	\label{thm:p-upper-bound-perfect-recovery}
	Let $\zeta_1,\zeta_2$ be defined as in Corollary~\ref{cor:null-space-constant-upper-bound-different-functions}. For any given $K_0\ge 1$, with matrix $\bvec{\Phi}$ satisfying $\delta_{2K_0}<1$, the following holds:
	\begin{enumerate}
		\item  For a given $K(K\le K_0)$, $P_f$ can recover any $K-$sparse vector under the following conditions: for $f$ as functions $1,2$ in Table~\ref{tab:scale-functions-examples}, $p\le p_{K,1}$, and for the function $3$ in Table~\ref{tab:scale-functions-examples}, $p_1\in(0,1),\ p_2\ge p_1,\ p_2\le p_{K,2}$ where \begin{align}
		\label{eq:p-upper-bound-perfect-recovery-exp-log}
		p_{K,1} & = \max\left\{p\in[0,1]\mid p\ln \zeta_1<\ln\left(\frac{2K_0-K+1}{K}\right)\right\},\\
		\label{eq:p-upper-bound-perfect-recovery-continuous-mixed-norm}
		p_{K,2} & = \max\left\{p\in[0,1]\mid p\ln \zeta_2<\ln\left(\frac{2K_0-K+1}{K}\right)\right\}.
		\end{align}
		\item For a given $p\in[0,1],p_1\in(0,1),p_2\in(p_1,1]$, $P_f$ can recover any $K-$sparse vector ($K\le K_0$) under the following conditions: for functions $1,2$ in Table~\ref{tab:scale-functions-examples}, $K\le K_1$, and and for function $3$ in Table~\ref{tab:scale-functions-examples}, $K\le K_2$, where \begin{align}
		\label{eq:K-upper-bound-perfect-recovery-exp-log}
		K_1 & = \max\left\{K\in[1,K_0],\ K\in \mathbb{Z}^+\mid\frac{K}{L'}\left(\indicator{\zeta_1\le 1} +\indicator{\zeta_1> 1}\zeta^p_1\right)<1\right\}, \\
		\label{eq:K-upper-bound-perfect-recovery-continuous-mixed-norm}
		K_2 & = \max\left\{K\in[1,K_0],\ K\in \mathbb{Z}^+\mid \frac{K}{L'}\left(\zeta_1^{p_1}\indicator{\zeta_2\le 1}+\indicator{\zeta_2>1}\zeta_2^{p_2}\right)<1\right\}.
		\end{align}   	
	\end{enumerate}
\end{thm}
\begin{proof}
	The proof of the first part follows by plugging in the expressions of $\zeta_1,\ \zeta_2$ in the expressions for $\gamma(J,\bvec{\Phi},K)$, and requiring that $\gamma(J,\bvec{\Phi},K)<1$. When $\zeta_1,\zeta_2\le 1$, the LHS of the bounds~\eqref{eq:p-upper-bound-perfect-recovery-exp-log} and~\eqref{eq:p-upper-bound-perfect-recovery-continuous-mixed-norm} are negative and so trivially satisfied. Thus, these bounds captures the maximum range of $p,p_2$ for all values of $\zeta_1,\zeta_2$.
	
	For the proof of the second part, again we just need to ensure that, for a fixed $p,p_2$, $K$ is chosen in such a way that $\gamma^\star(J,\bvec{\Phi},K)<1$. Plugging in the expressions of $\gamma^\star(J,\bvec{\Phi},K)$ from~\eqref{eq:null-space-constant-upper-bound-expression-exp-or-log} and~\eqref{eq:null-space-constant-upper-bound-expression-continuous-mixed-norm}, the bounds~\eqref{eq:K-upper-bound-perfect-recovery-exp-log} and~\eqref{eq:K-upper-bound-perfect-recovery-continuous-mixed-norm} follow.
\end{proof}

%{\color{red} I have to now discuss some of the implications of this Theorem for different functions. Then I have to find an upper bound on the RIC in terms of $K,g$, and mark that as a Theorem. In that theorem I have to discuss the different cases for different functions: concave exp, log, mixed continuous norm. The bounds for concave exp and log will be the same, as their scale functions are the same. After this, I have to make some propositions or corollaries about the limit of $K_0$ for which exact retrieval can be performed.}
%
\section{Numerical results}
\label{sec:numerical-results}
In this section we numerically study the recovery performance of the problem $P_f$ using the different sparsity promoting functions $f(x)=\abs{x}^p,\ln(1+\abs{x}^p),1-e^{-\abs{x}^p}$, which will now be referred to as $f_a(x),f_b(x)$ and $f_c(x)$, respectively, in the rest of this section. The parameters of interest that are varied to study the comparative performances are the sparsity $K$ of the problem, and the maximum elasticity, which for all these functions, is simply $p$ (See Table~\ref{tab:scale-functions-examples}). For the $l_p$ minimization problem, as the parameter $p$ is made smaller, the corresponding $l_p$ minimization problem becomes more nonconvex and therefore harder to solve~\cite{lai2013improved}, although with smaller number of measurements~\cite{chartrand2008restricted}. Similar implications can also be expected to hold true for the more general sparsity promoting functions $f_a,f_b,f_c$, which would imply the significance of conducting experiments by varying the maximum elasticity $p$. However, a detailed study of how the maximum elasticity affects the hardness of numerically solving $P_f$ with these sparsity promoting functions seems to be rather difficult and beyond the scope of this paper.   
\subsection{Example of recovery properties of $P_f$}
\label{sec:example-recovery-properties-Pf}
%
%In this section we consider exact recovery of $1$-sparse vectors by solving the problem $P_f$ and compare the relative performances when $f(x)=\abs{x}^p$, $f(x)=\ln(1+\abs{x}^p)$ or $f(x)=1-e^{-\abs{x}^p}$ is used. 
%numerically analyze the properties in Theorems~\ref{thm:null-space-constant-upper-bound} and~\ref{thm:delta-upper-bound-perfect-recovery} and compare between the cases when .
%
\begin{figure}[t!]
	\centering
	\begin{subfigure}{0.3\textwidth}
		\centering
		\includegraphics[width=2.25in, height=1.5in]{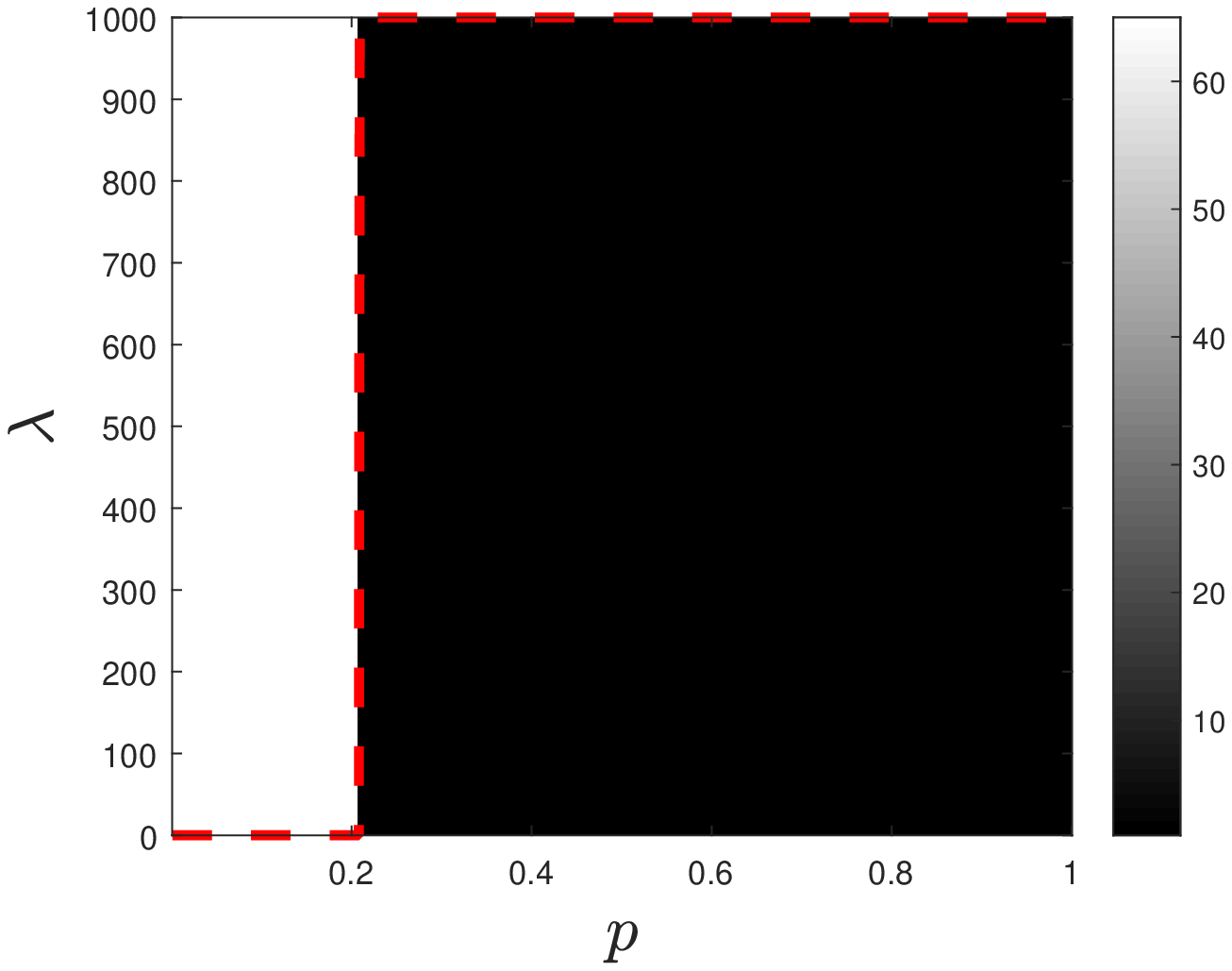}
		\caption{$f(x)=\abs{x}^p$}
		\label{fig:recovery-vs-lambda-p-lp}
	\end{subfigure}
	\begin{subfigure}{0.3\textwidth}
		\centering
		\includegraphics[width=2.25in, height=1.5in]{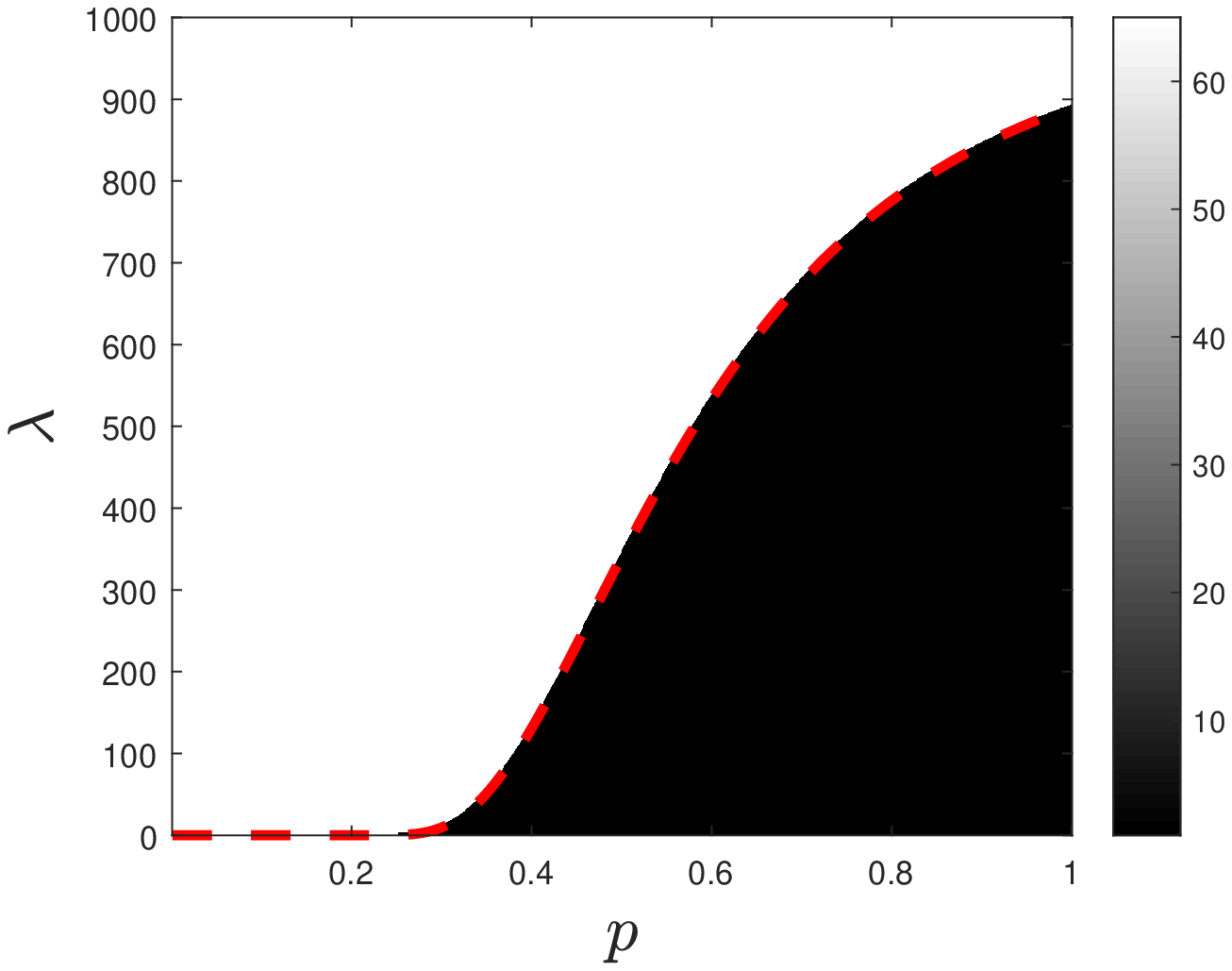}
		\caption{$f(x)=\ln(1+\abs{x}^p)$}
		\label{fig:recovery-vs-lambda-p-log}
	\end{subfigure}
	\begin{subfigure}{0.3\textwidth}
		\centering
		\includegraphics[width=2.25in, height=1.5in]{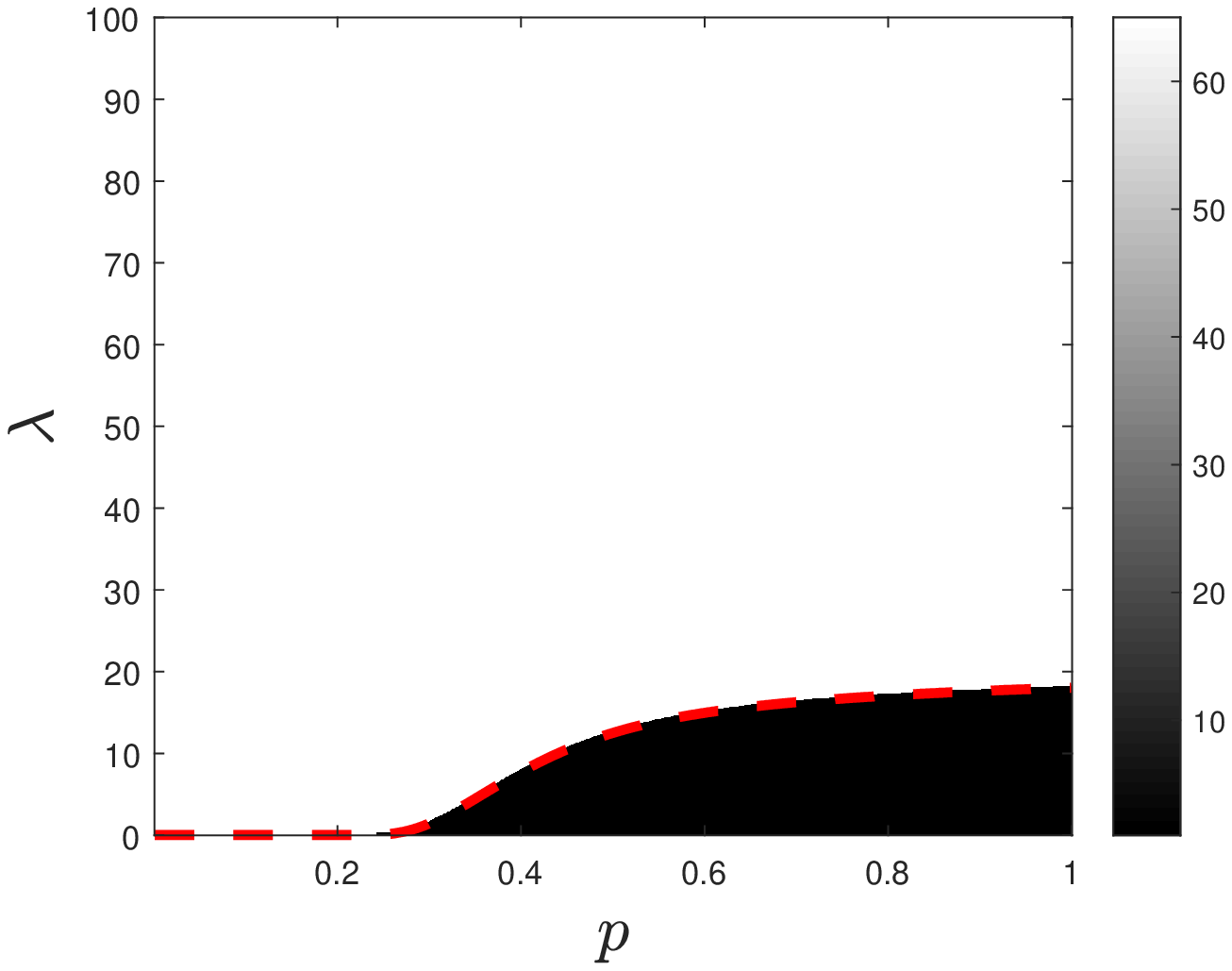}
		\caption{$f(x)=1-e^{-\abs{x}^p}$}
		\label{fig:recovery-vs-lambda-p-exp}
	\end{subfigure}
	\caption{Exact recovery of $1$-sparse vector solving the problem $Q_1$ for different $\lambda$ and $p$ for different sparsity promoting functions. For each pair $(\lambda,p)$, a white color represents a perfect recovery whereas a black color represents otherwise.}
	\label{fig:recovery-vs-lambda-p}
\end{figure}
In the first experiment we numerically study the recovery capabilities of the $J$-minimization (that is the problem $P_f$) by considering exact recovery of $1$-sparse vector using the functions $f_a,f_b$ and $f_c$. For any $1\le i\le N$, we study the following version of the problem $P_f$, which we will call $Q_i$: \begin{align}
\min_{\bvec{x}\in \real^N}\sum_{j=1}^N f(x_j),\ & \mathrm{s.t.}\ \bvec{y}=\lambda\bvec{\Phi}\bvec{e}_i,
\end{align}
where $\bvec{x}=\lambda\bvec{e}_i$, $\lambda\ne 0$ is a real number and $\bvec{e}_i\in \real^N$ is the $i^\mathrm{th}\ (1\le i\le N)$ standard basis vector with $[\bvec{e}_i]_i=1$ and $[\bvec{e}_i]_j=0,\ \forall j\ne i$. We also assume that the matrix $\bvec{\Phi}$ has full row rank $M$ and that $N=M+1$, so that the there exists a vector $\bvec{v}$ such that $\mathcal{N}(\bvec{\Phi})=\spn{\bvec{v}}$. Let us assume that $\bvec{v}=[v_1\ v_2\cdots v_N]^t$ and that the entries $v_i$ are ordered according to magnitude as $\abs{v_{\pi(1)}}\ge \cdots\ge \abs{v_{\pi(N)}}$, where $\pi$ is some permutation of the set $\{1,\cdots,N\}$. 

Now, consider solving the problem $Q_{i},1\le i\le N$. Clearly, any element of the constraint set satisfies $\bvec{x}=\lambda \bvec{e}_i + u\bvec{v}$, for some $u\in \real$. Therefore, to find the optimal $\bvec{x}$, one needs to find $u$ which minimizes the function $h_i(u)=f(\lambda+uv_i)+\sum_{j\ne i}f(uv_j)$. Observe that since the function $f$ is even, it is sufficient to assume $\lambda>0$ and consider the function $h_i(u)=f(\lambda+u\abs{v_i})+\sum_{j\ne i}f(u\abs{v_j})$. If $v_i=0$, then it is obvious that the minimizer of $h_i$ is $0$ and therefore the solution to $Q_i$ yields exact recovery. Therefore, we only consider the case $\abs{v_i}>0$. In this case, note that the function $h_i(u)$ is piece-wise concave in the intervals $(-\infty, -\lambda/\abs{v_i}),\ (-\lambda/\abs{v_i},0),(0,\infty)$. Since concave functions assume minimum at endpoints, and since the function $h_i(u)$ is decreasing in $(-\infty, -\lambda/\abs{v_i})$ and increasing in $(0,\infty)$, finding minimum of the function $h_i(u)$ is equivalent to finding $\min\{\sum_{j\ne i}f(\lambda \abs{v_j}/\abs{v_i}),\ f(\lambda)\}$. Therefore solution of the problem $Q_i$ yields exact recovery if and only if the solution is $u=0$, or equivalently, $\sum_{j\ne i}f(\lambda \abs{v_j}/\abs{v_i})\ge  f(\lambda)$. Clearly, when $i\ne \pi(1),$ $\sum_{j\ne i}f(\lambda \abs{v_j}/\abs{v_i})\ge f(\lambda\abs{v_{\pi(1)}}/\abs{v_i})\ge f(\lambda)$. Therefore, the solution of all the problems $Q_i,\ i\ne \pi(1)$ yield exact recovery. Furthermore, if $\sum_{j\ne 1}\abs{v_{\pi(j)}}\ge \abs{v_{\pi(1)}}$, we have $\sum_{j\ne i}f(\lambda \abs{v_{\pi(j)}}/\abs{v_{\pi(1)}})\ge f\left(\lambda\frac{\sum_{j\ne 1}\abs{v_{\pi(j)}}}{\abs{v_{\pi(1)}}}\right)\ge f(\lambda)$. Therefore, it only remains to see how the different functions perform while recovering exact solutions of the problem $Q_{\pi(1)}$ when $\sum_{j\ne 1}\abs{v_{\pi(j)}}< \abs{v_{\pi(1)}}$. To avoid confusion, we refer problem $Q_{\pi(1)}$ with such constraint on the corresponding vector spanning the null space as $Q_{\pi(1)}'$.

In order to study the solutions of the problem $Q'_{\pi(1)}$ we consider the simulation setup similar to the one used in~\cite[Section VII-A]{gu2018sparse-lp-bound}, and use the following sensing matrix: \begin{align}
\label{eq:example-measurement-matrix}
\bvec{\Phi} = \sqrt{\frac{2}{101.01}}\begin{bmatrix}
0.1 & 0 & 10\\
1 & -10 & 0
\end{bmatrix}.
\end{align} As discussed in Sec VII-A of~\cite{gu2018sparse-lp-bound}, for this matrix $K_0=1, \delta_{2K_0}=0.9998<1$, so that $P_0$ can recover any $1-$sparse signal. The null space of this matrix is spanned by $\bvec{v}=[100\ 10\ -1]^t$, so that, according to the discussion in the last paragraph, $\pi(1)=1$. Therefore, the corresponding $h_1$ function is \begin{align}
\label{eq:h_1-expression}
h_1(u) & = f(\lambda+100u)+f(10u)+f(u).
\end{align}
%We now study how the exact recovery is affected by solving problem $Q_1$ by considering the the functions $f_a,f_b$ and $f_c$ and changing them by changing $p$ and then studying conditions for which the bound $f(\lambda)\le f(\lambda/10)+f(\lambda/100)$ is satisfied for different $\lambda>0$.
%when different sparsity promoting functions, such as those in Table~\ref{tab:scale-functions-examples} are used. Each of these functions is specified by the parameter $p$ which, apart from being the maximum elasticity of the three sparsity promoting functions under consideration, is also important in the sense that as $p$ becomes closer to $0$ it can be expected that the associated problem $P_f$ becomes harder to solve. Intuitively this might be understood as an extension of the fact that the problem $P_f$ with $f(x)=\abs{x}^p$ becomes increasingly difficult to solve as $p\downarrow 0$, although a rigorous investigation is difficult and out of scope of the current research. Therefore, although it might be the case that a small value of $p$ might be theoretically more interesting possessing to better recovery performance, a somewhat larger value of $p$, i.e., closer to $1$, might be more desirable from a computational point of view. Therefore, by changing the values of $p$, we consider different functions and plot the recovery performances of these functions by 

In Fig.~\ref{fig:recovery-vs-lambda-p}, exact recovery of a $1-$sparse vector by solving the problem $Q_{1}'$ with the measurement matrix of~\eqref{eq:example-measurement-matrix} is investigated using the sparsity promoting functions $f_a,f_b,f_c$. Since each of the functions $f_a,f_b,f_c$ is specified by their maximum elasticity $p$, for each pair of $\lambda$ and $p$, whether the solution of $Q_1'$ yields exact recovery is indicated by coloring the tuple $(\lambda,p)$ either in white if exact recovery occurs or in black otherwise. From Figs.~\ref{fig:recovery-vs-lambda-p-lp},~\ref{fig:recovery-vs-lambda-p-log} and~\ref{fig:recovery-vs-lambda-p-exp}, it can be seen that irrespective of the value of $\lambda$, for all the sparsity promoting functions, exact recovery for \emph{all} values of $\lambda$ occurs when $p\in (0,p_0)$ where $p_0$ is somewhat larger than $0.2$. Interestingly, for values of $p$ larger than this threshold, while exact recovery is not possible for \emph{any} value of $\lambda$ in the range $(0,1000)$ when $f_a$ is used, using the functions $f_b,f_c$, exact recovery is still possible when $\lambda$ is larger than the values indicated by the respective dashed curves in Figs.~\ref{fig:recovery-vs-lambda-p-log} and~\ref{fig:recovery-vs-lambda-p-exp}. Since the requirement that the bound $f(\lambda)\le f(\lambda/10)+f(\lambda/100)$ is satisfied for all $\lambda>0$, can be seen to be equivalent to requiring the corresponding $J-$NSC $< 1$ and since it can be shown that for any given $p\in(0,1]$, the $J-$NSC for $f_a,f_b,f_c$ are the same (see Appendix~\ref{sec:appendix-evaluating-nsc-null-space-spanned-by-single-vector}), the simulations in Fig.~\ref{fig:recovery-vs-lambda-p} imply that while for the function $\abs{x}^p$, exact recovery of any $1$-sparse vector is impossible by solving $Q_1'$ when the corresponding $J-$NSC $\ge 1$, for the functions $\ln(1+\abs{x}^p)$ and $1-e^{-\abs{x}^p}$, it is still possible for large enough $\lambda$.

\subsection{Comparison of NSC bounds with true NSC}
\label{sec:numerical-comparison-nsc-bounds}
\begin{figure}[t!]
	\centering
	%	\hspace{-1cm}
	\begin{subfigure}{.275\textwidth}
		%		\centering
		\includegraphics[height = 1.5in, width = 1.75in]{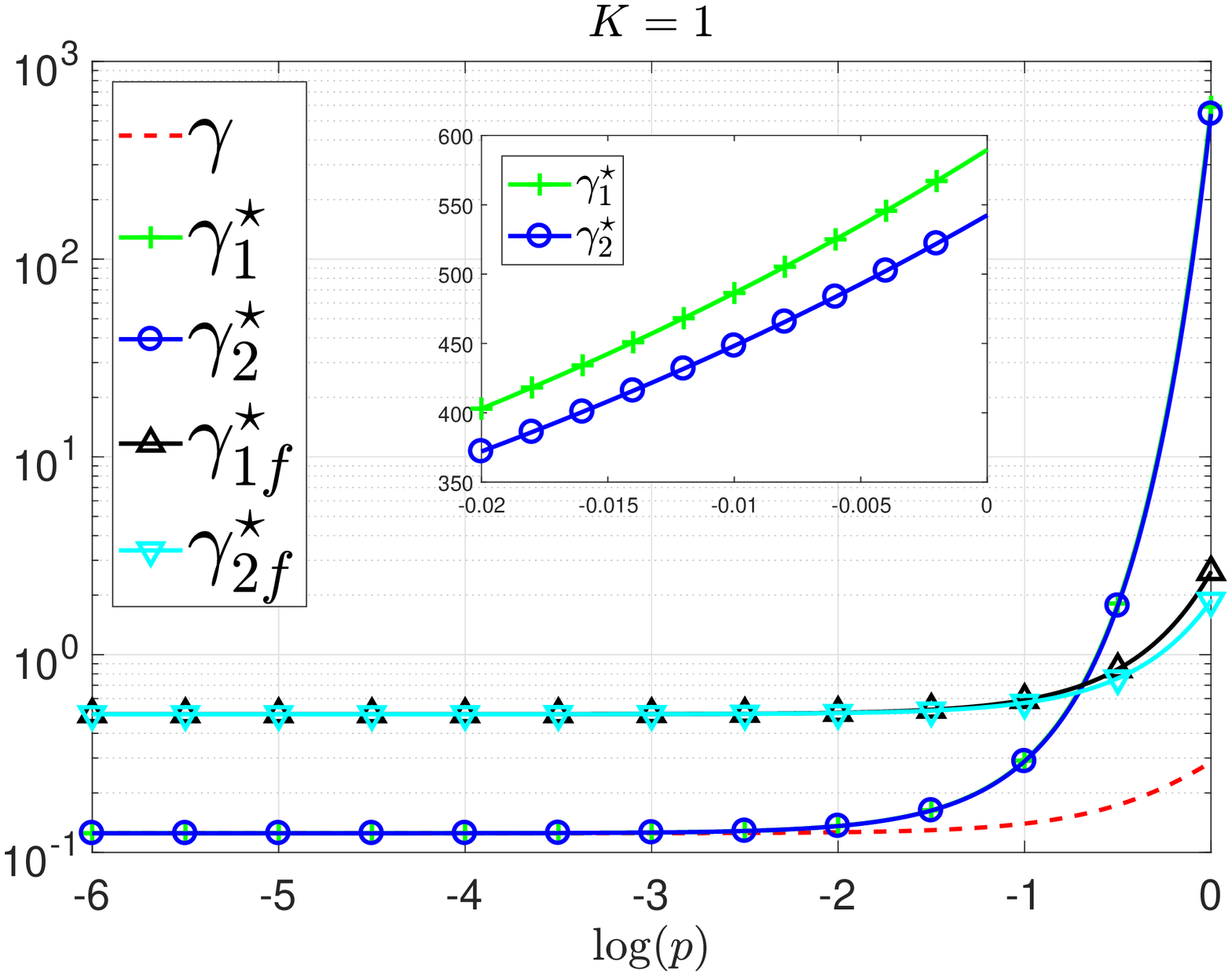}
	\end{subfigure}%
	\begin{subfigure}{.275\textwidth}
		%		\centering
		\includegraphics[height = 1.5in, width = 1.75in]{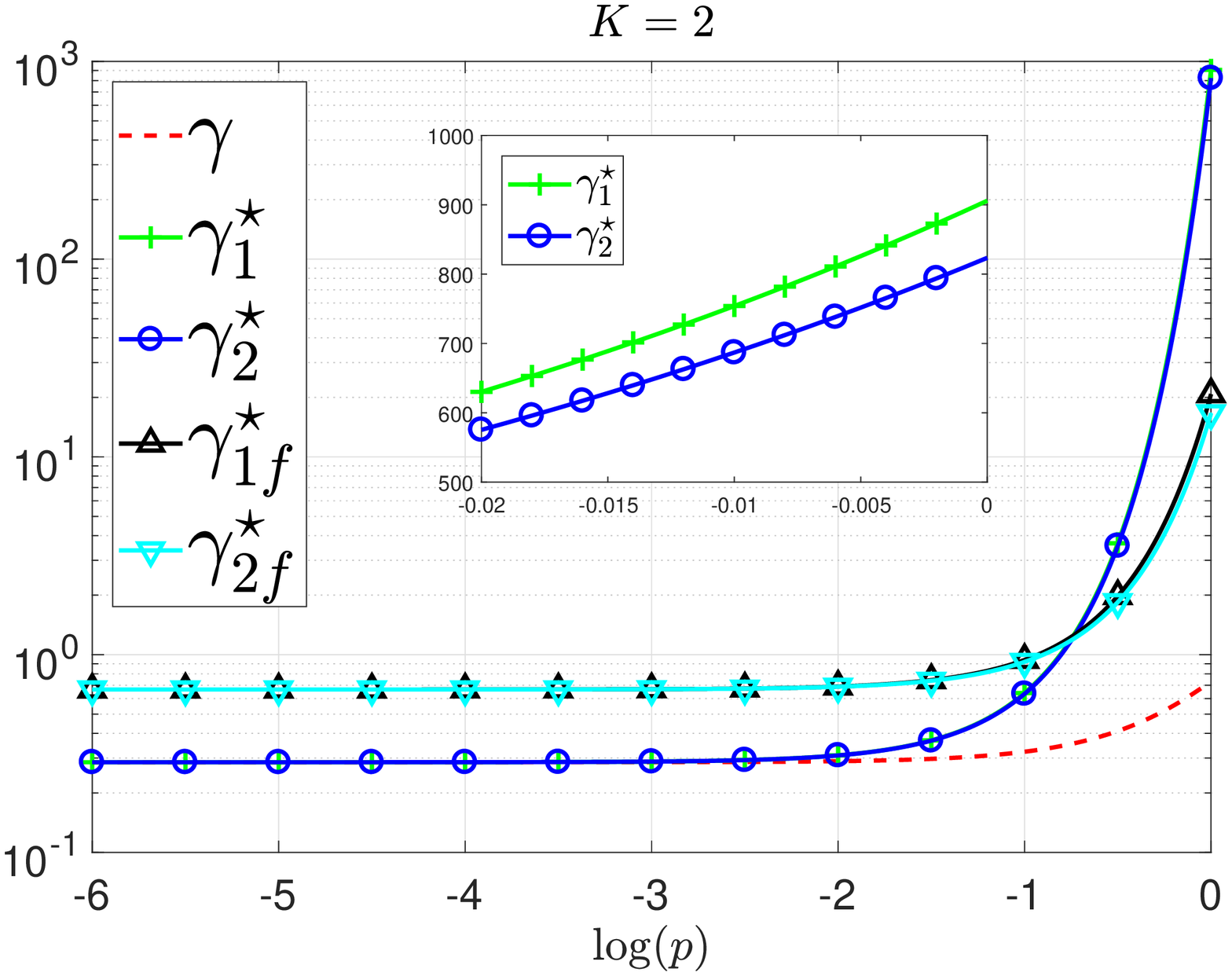}
	\end{subfigure}
	%
	
	%	\hspace{-1cm}
	\centering
	\begin{subfigure}{.275\textwidth}
		%		\centering
		\includegraphics[height = 1.5in, width = 1.75in]{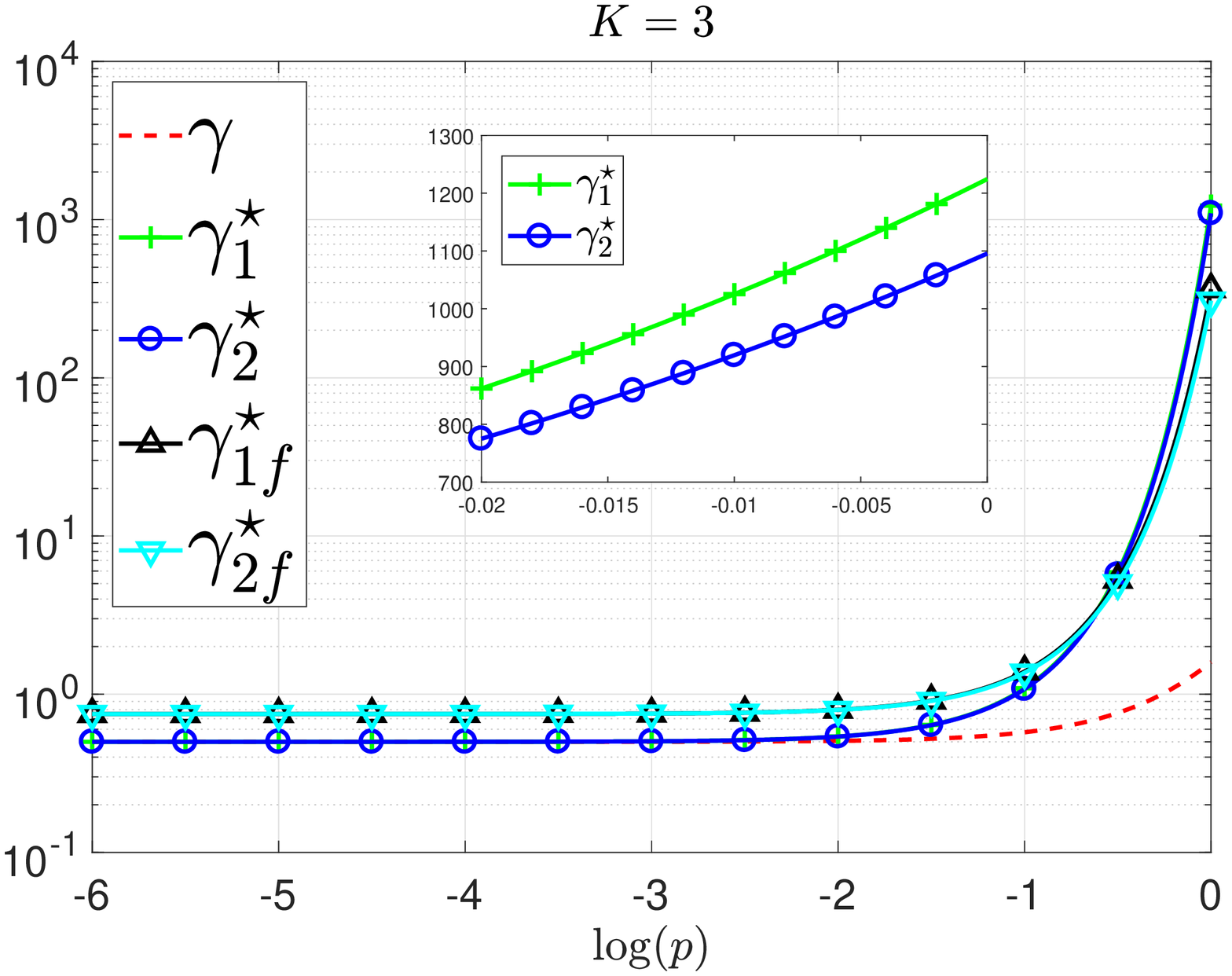}
	\end{subfigure}%
	\begin{subfigure}{.275\textwidth}
		%		\centering
		\includegraphics[height = 1.5in, width = 1.75in]{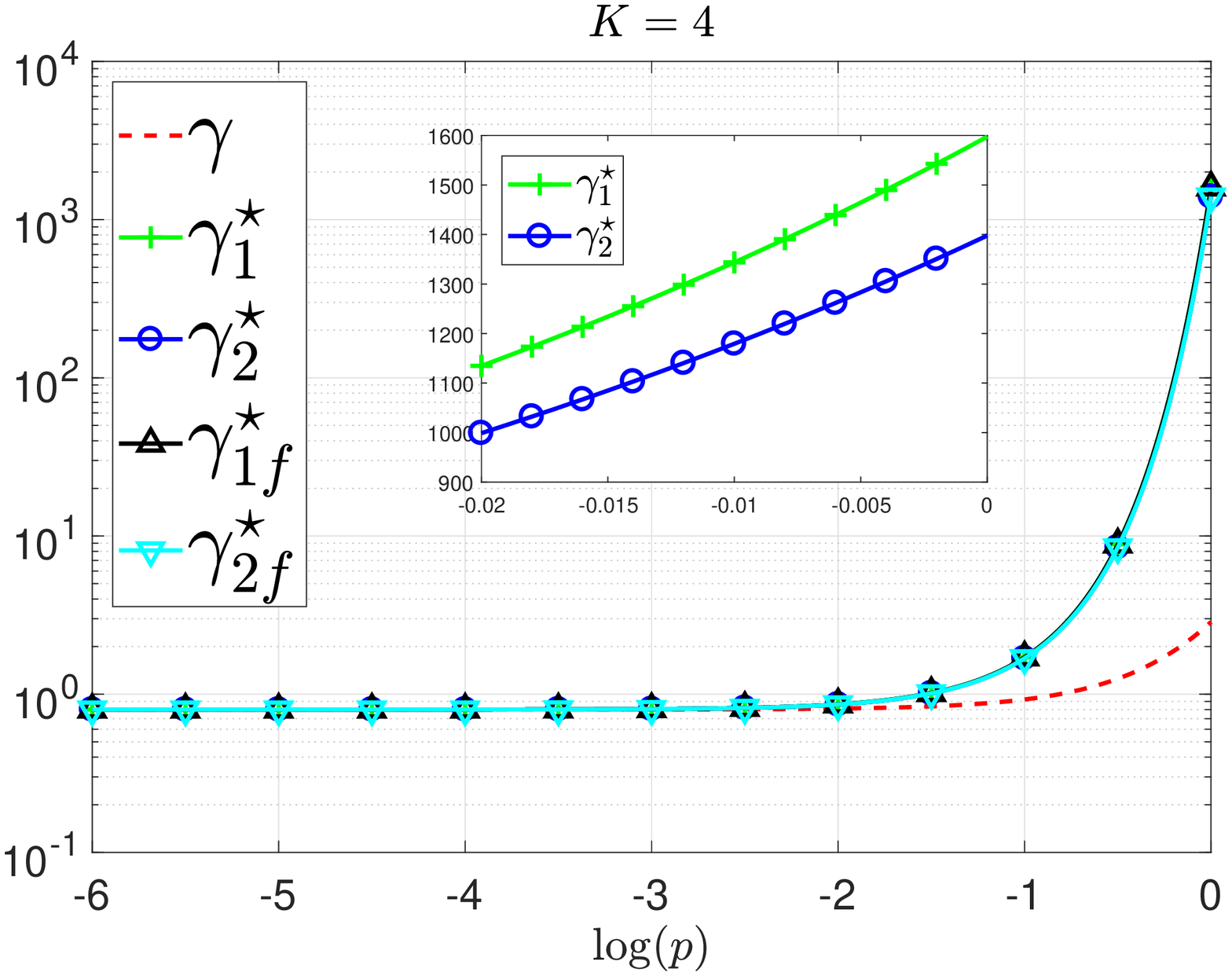}
	\end{subfigure}
	
	\caption{Variation of null space constant $\gamma$ and its bounds $\gamma^\star,\gamma_{1f}^\star,\gamma_{2f^\star}$ for  different functions and different $K$.}
	\label{fig:variation-nsc}
\end{figure}
In the second experiment, the bounds on NSC for the different functions are compared with the true NSC for a randomly generated measurement matrix. 

We consider the experimental setup used in Section VII-B of~\cite{gu2018sparse-lp-bound} and generate a random matrix $\bvec{\Phi}\in \real^{8\times 9}$ with standard normal entries and then normalize the columns of the matrix. The RIC's of the matrix are $\delta_2 = 0.5617,\ \delta_4 = 0.9034,\ \delta_6 = 0.9781,\ \delta_8 = 0.9999$. Clearly, $K_0 = 4$, and $\mathcal{N}(\bvec{\Phi})$ is spanned by a single vector. In this case, the $J-$NSC for the functions $f_a,f_b,f_c$ can be shown to be all identical. See Appendix~\ref{sec:appendix-evaluating-nsc-null-space-spanned-by-single-vector} for a proof of this fact.
% to  $\frac{\sum_{i=1}^{K}(z_i^+)^p}{\sum_{i=K+1}^{N}(z_i^+)^p}$, where $\bvec{z}^+$ is the positive rearrangement of $\bvec{z}$, i.e., $\bvec{z}^+$ has entries as the absolute values of entries of $\bvec{z}$, sorted in descending order. 

In Fig~\ref{fig:variation-nsc}, for different values of $K$, the variation of the actual NSC and the bounds~\eqref{eq:null-space-constant-upper-bound-expression-exp-or-log}, and~\eqref{eq:null-space-constant-upper-bound-expression-continuous-mixed-norm} are plotted against the maximum elasticity of the functions ($p$) with either knowing only $\delta_8$, or with knowing all the RICs i.e., $\delta_2,\ \delta_4,\ \delta_6$ and $\delta_8$. In the plots, the true NSC is referred to as $\gamma$; with only knowing $\delta_8$, the NSC bound obtained from~\eqref{eq:null-space-constant-upper-bound-expression-continuous-mixed-norm} for $f_a$ is referred to as $\gamma_1^\star$, while for functions $f_b,f_c$, the bound obtained from~\eqref{eq:null-space-constant-upper-bound-expression-exp-or-log} is referred to as $\gamma_2^\star$. With the knowledge of $\delta_2,\delta_4,\delta_6$ and $\delta_8$, the bound for $f_a$ is referred to as $\gamma_{1f^\star}$, while for functions $f_b,f_c$ the corresponding bound is referred to as $\gamma_{2f}^\star$. All the plots reveal that for small $p$, $\gamma_2^\star$ is the best and for large $p$ $\gamma_{2f}^\star$ is the best bound. For small $p$ for all values of $K$, $\gamma_1^\star$ and $\gamma_2^\star$ are almost equal to the actual NSC $\gamma$, whereas, for large $K$ ($K=4$), all the bounds are almost equal for small $p$. This indicates that for highly sparse systems, for small $p$, the bounds derived in Corollary~\ref{cor:null-space-constant-upper-bound-different-functions} become almost equal to the actual NSC.

\subsection{Recovery conditions for different values of $K$ and $p$}
\label{secv:numerical-recovery-conditions-K-p}
\begin{figure}[t!]
	\centering
	\includegraphics[height = 1.6in, width = 3.2in]{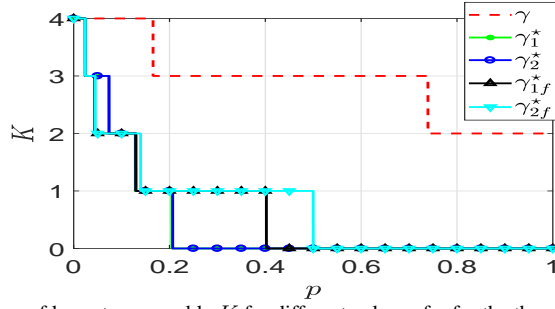}
	\caption{Variation of largest recoverable $K$ for different values of $p$ for the three different functions $f$.}
	\label{fig:K-variation-vs-p}
\end{figure}
In Figure~\ref{fig:K-variation-vs-p} the ranges of the maximum elasticity $p$ for which a particular sparsity $K$ is recoverable, are plotted so as to ensure that $\gamma,\gamma_1^\star,\gamma_2^\star,\gamma_{1f}^\star$ and $\gamma_{2f}^\star$ are $<1$ respectively. From the plot it can be seen that although none of the bounds $\gamma_1^\star,\gamma_2^\star,\gamma_{1f}^\star$ and $\gamma_{2f}^\star$ has as broad recoverable range of $p$ as that of the true NSC, the bounds $\gamma_2^\star$ and $\gamma_{2f}^\star$ have broader ranges than $\gamma_1^\star$ and $\gamma_{1f}^\star$ respectively. Moreover, for low values of $p$, the bounds with knowledge of only $\delta_8$ have broader recovery range than the bounds with full information of all the RICS and demonstrate the capability of recovering larger $K$. On the other hand, for large $p$, the recovery range for the bounds with full information about all the RICs are better than the one with knowledge of only $\delta_8$. However, for large $p$, the recoverable sparsity order becomes smaller. 
\subsection{Comparison of upper bounds~\eqref{eq:delta-upper-bound-perfect-recovery-exp-log} and~\eqref{eq:delta-upper-bound-perfect-recovery-continuous-mixed-norm}}
\label{sec:numerical-comparison-delta-upper-bounds}
\begin{figure}[t!]
	\centering
	\begin{subfigure}{0.25\textwidth}
		\centering
		\includegraphics[height = 1.5in, width = 1.75in]{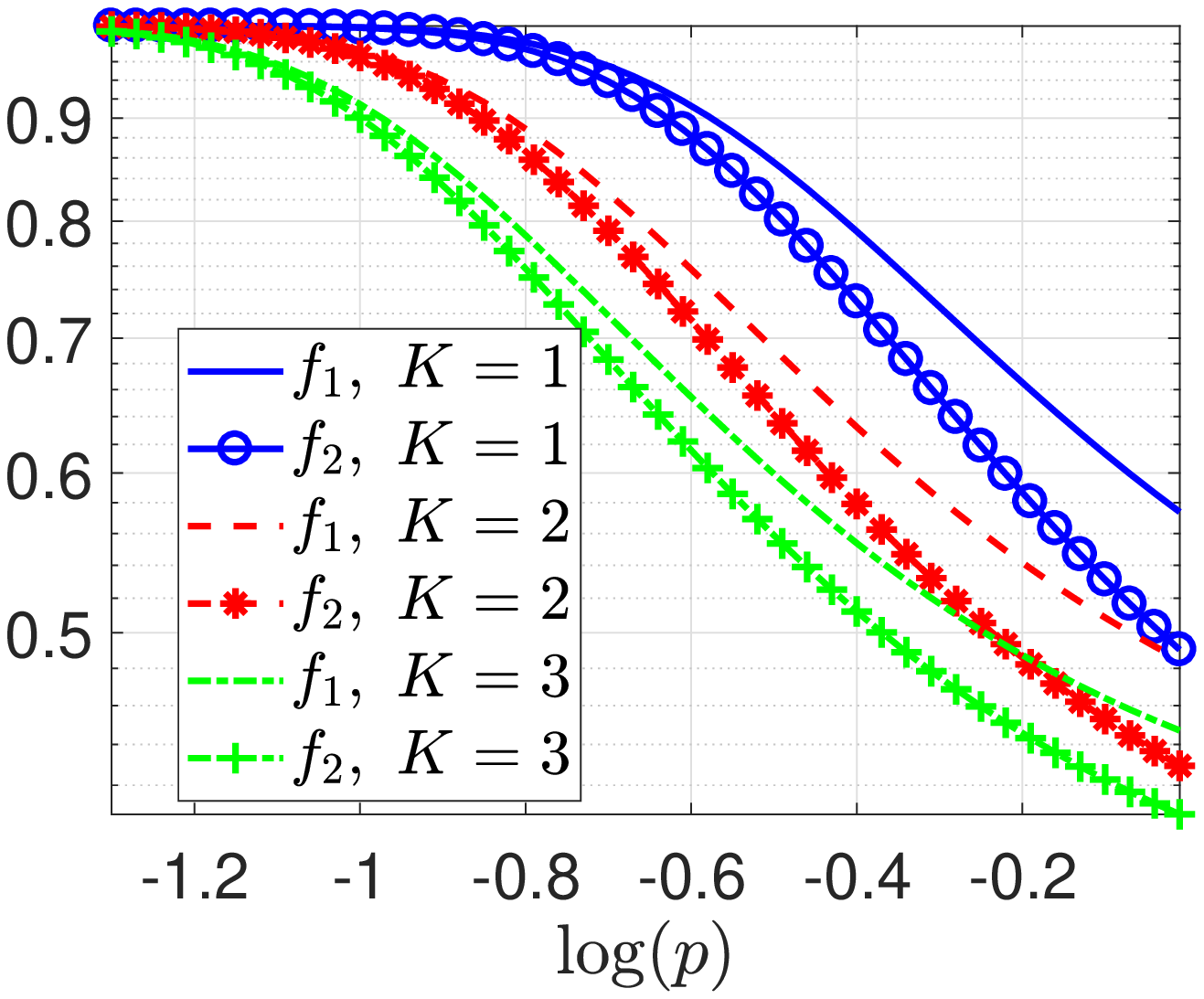}
		\caption{$K_0=K$}
		\label{fig:delta-upper-bounds-comparison-K0=K}
	\end{subfigure}
	\begin{subfigure}{0.25\textwidth}
		\centering
		\includegraphics[height = 1.5in, width = 1.75in]{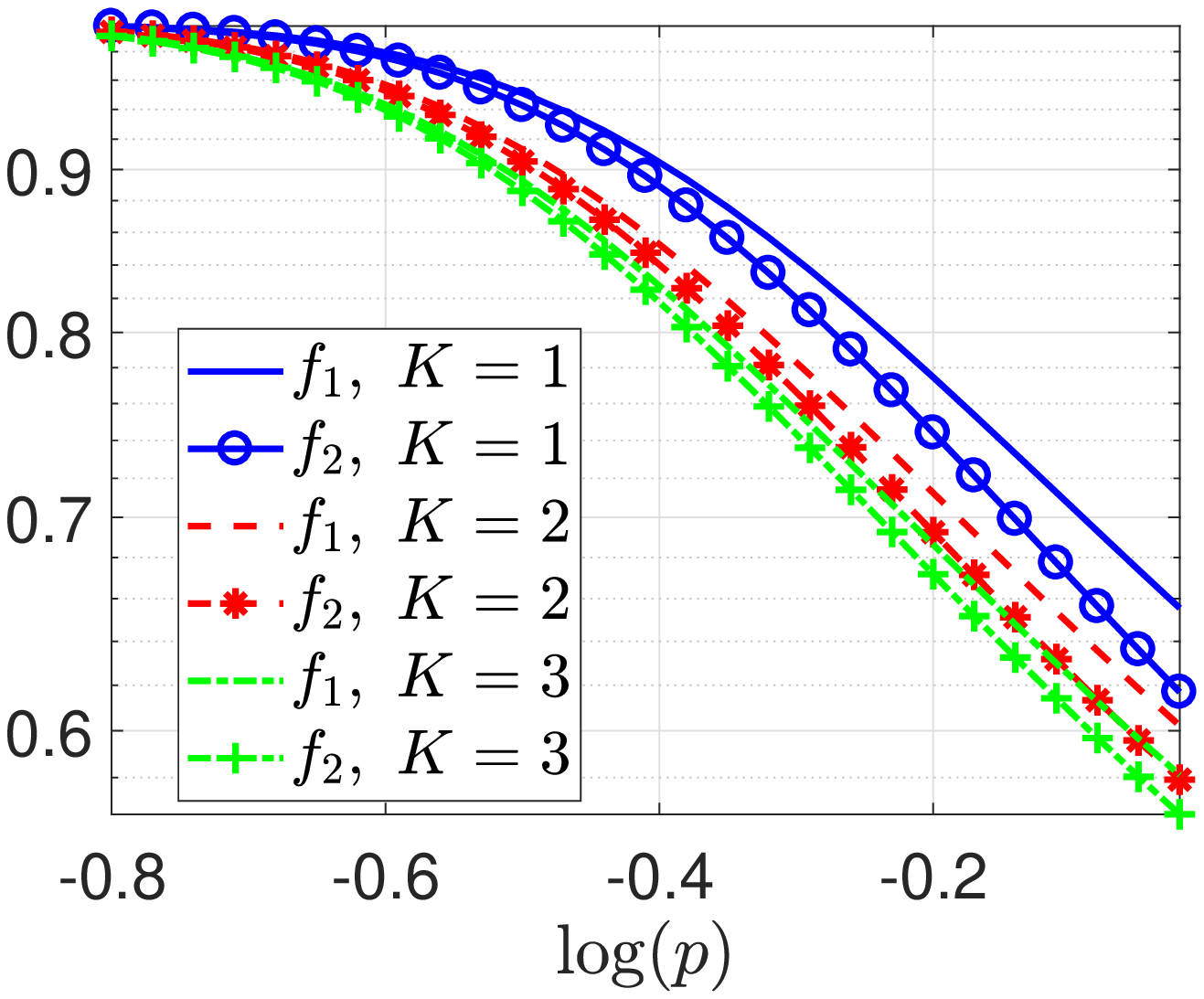}
		\caption{$K_0=2K$}
		\label{fig:delta-upper-bounds-comparison-K0=2K}
	\end{subfigure}
	\caption{Comparison of the upper bounds~\eqref{eq:delta-upper-bound-perfect-recovery-exp-log} and~\eqref{eq:delta-upper-bound-perfect-recovery-continuous-mixed-norm} on $\delta_{2K_0}$ for different values of $K,K_0$.}
	\label{fig:delta-upper-bounds-comparison}	
\end{figure}
In this section we plot the bounds $f_1(K,K_0,p)$~\eqref{eq:delta-upper-bound-perfect-recovery-exp-log} and $f_2(K,K_0,p_2)$ (with $p_2=p$)~\eqref{eq:delta-upper-bound-perfect-recovery-continuous-mixed-norm} against the maximum elasticity $p$ in the Fig.~\ref{fig:delta-upper-bounds-comparison} for different values of $K$. In Fig.~\ref{fig:delta-upper-bounds-comparison-K0=K} we use $K_0=K$ and in Fig.~\ref{fig:delta-upper-bounds-comparison-K0=2K} we use $K_0=2K$. From both the figures, it is evident that the upper bound $f_2$ is smaller than $f_1$ for a fixed $K,K_0,p$. Furthermore, for fixed $p$, both $f_1,f_2$ can be seen to decrease with increasing $K,K_0$, that is, as $K,K_0$ increase, the recovery bound on $\delta_{2K_0}$ decreases for all the functions, which is expected as recovering vectors of larger sparsity requires the measurement matrix to be closer to being ``unitary''. 
%Moreover, when $K_0=K$, it can also be observed from Fig.~\ref{fig:delta-upper-bounds-comparison-K0=K} that for $p>10^{-0.2}\approx 0.63$, $f_1(3,3,p)>f_2(2,2,p)$, so that the RIC bound on $\delta_6$ for retrieving a $3-$sparse vector (by solving the $J$-minimization problem) using the functions $\ln(1+\abs{x}^p)$ and $1-e^{-\abs{x}^p}$ is larger than the RIC bound on $\delta_{4}$ for retrieving a $2$-sparse vector using the function $\abs{x}^p$. This further corroborates the discussion in the paragraph following~\ref{thm:delta-upper-bound-perfect-recovery} on the usefulness of using alternate sparsity promoting functions other than the $l_p$.
%%
\section{Conclusion}
\label{sec:conclusion}
This paper aims to study sparse recovery problems, termed as the $J$-minimization problems, using general sparsity promoting functions using RIP and NSP. In order to do that the analysis technique of~\cite{gu2018sparse-lp-bound}, which focused on studying the performance of sparse recovery of $l_p$ functions, is generalized to the case of general sparsity promoting functions by introducing the notion of elasticity of a function and scale function. Thereafter, bounds on the general $J$-NSC is derived and is used to find concrete bounds for three particular sparsity promoting functions, $f(x)=\ln(1+\abs{x}^p),1-e^{-\abs{x}^p},\int_0^1 \abs{x}^p d\nu$. These specialized results are further used to find concrete bounds on the RIC that ensures sparse recovery using these functions, and are thereafter utilized to find bounds on maximum recoverable sparsity using these functions. Finally, it is numerically demonstrated that the $J-$minimization problem using $f(x)=\ln(1+\abs{x}^p),1-e^{-\abs{x}^p}$ has interesting theoretical recovery properties which are very competitive to the the $l_p$ minimization problem and in some cases more advantageous. Furthermore, the actual NSC, as well as the true maximum recoverable sparsity for these functions are evaluated and compared with the derived bounds for an example measurement matrix and the resulting implications are discussed. Several interesting questions, for example how the maximum elasticity affects the recovery capability of a general sparsity promoting functions, in terms of hardness of solving the corresponding $J$-minimization problem, as well as convergence rate of relevant solution techniques, can be considered as topics for future research.
%
%\section{Examples of functions satisfying the properties in~\ref{def:f-properties}}
%
%\label{sec:examples-f-function}

%so that $f'(x)/f''(x)=-(1+x)$, which is convex.   
%
 
%To verify the property that $f'/f''$ is convex, we proceed as below. 

%We take $x>0$, w.l.o.g. Let us denote $,\forall k\ge 1, \mu_k=\expect{p}{p^k x^p}$, so that $\mu_k'=\frac{d\mu_k}{dx}=\expect{p}{p^{k+1}x^{p-1}}=\mu_{k+1}/x$. Therefore, $f'(x)=\mu_1/x,\ f''(x)=\left(\mu_2-\mu_1\right)/x^2$. Let $h(x)=f'(x)/f''(x)=x\mu_1/(\mu_2-\mu_1)$. A straightforward calculation yields that, \begin{align}
%h'(x) & =\frac{\mu_2+\mu_1}{\mu_2-\mu_1}-\frac{\mu_1(\mu_3-\mu_2)}{(\mu_2-\mu_1)^2},\\
%xh''(x) & =\frac{2(\mu_2^2-\mu_1\mu_3)}{(\mu_2-\mu_1)^2}+\frac{2(\mu_3-\mu_2)^2}{(\mu_2-\mu_1)^3}-\frac{\mu_4-\mu_3}{(\mu_2-\mu_1)^2}
%\end{align}
%
\section*{Appendix}
\section{Proof of Lemma~\ref{lem:generalized-holder-type-bounds} }
\label{sec:appendix-proof-of-lemma-generalized-holder-type-bounds}
\subsection{Lower bound}
\label{sec:appendix-proof-of-lemma-generalized-holder-type-lower-bound}
First, we establish the lower bound in~\eqref{eq:main-holder-type-inequality-generalized-concave-function}. Before proceeding any further, we first observe that we can restrict ourselves to the case when $\bvec{a}\in \real_+^s$, since the function $f$ as well as $\opnorm{\cdot}{q}$ are even functions. More importantly, we observe that establishing the lower bound is equivalent to proving that $\sum_{j=1}^sf(a_j)\le sf\left(\frac{\opnorm{\bvec{a}}{q}}{s^{1/q}}\right)$, for $q\ge 1$, and $\bvec{a}\in \real_+^s$. However, the last statement is equivalent to showing that $\forall\ \xi>0$, $\opnorm{\bvec{a}}{q}\le s^{1/q}f^{-1}(\xi/s)\implies \sum_{j=1}^sf(a_j)\le \xi$. Hence, we fix any $\xi>0$, and solve the following optimization problem~\eqref{eq:equivalent-optimization-problem-lower-bound}: \begin{align}
\label{eq:equivalent-optimization-problem-lower-bound}
\max_{\bvec{a}} \sum_{j=1}^s f(a_j),\ \mathrm{s.t.}\bvec{a}\in \mathcal{C},
\end{align}
where $\mathcal{C}=\left\{\bvec{a}\in \real_+^s:\opnorm{\bvec{a}}{q}^q\le h(\xi)\right\}$, and $h(\xi)=s\left(f^{-1}(\xi/s)\right)^q$. We intend to show that the solution of the optimization problem~\eqref{eq:equivalent-optimization-problem-lower-bound} will yield a value of the objective function, which is $\le \xi$. 

To solve the optimization problem~\eqref{eq:equivalent-optimization-problem-lower-bound}, we first note that the objective function $\sum_{j=1}^s f(a_j)$ is concave whenever $\bvec{a}\in \real_+^s$. Also, the set $\mathcal{C}$ can be easily shown to be convex, since the function $\opnorm{\cdot}{q}$ is convex for $q\ge 1$. Hence, the optimizer of the problem~\eqref{eq:equivalent-optimization-problem-lower-bound} is its KKT point. 

To find the KKT point, we form the Lagrangian for problem~\eqref{eq:equivalent-optimization-problem-lower-bound}, \begin{align}
\label{eq:lagrangian-for-equivalent-optimization-problem}
\mathcal{L}(\bvec{a},\lambda) = \sum_{j=1}^sf(a_j)+\lambda\left(\opnorm{\bvec{a}}{q}^q-h(\xi)\right).
\end{align}
The KKT point $\bvec{a}^\star$ satisfies, for all $j=1,\cdots, s$, \begin{align}
\frac{\partial \mathcal{L}}{\partial a_j}|_{a_j=a_j^\star}=0, & 
\ \lambda\left(\opnorm{\bvec{a}^\star}{q}^q-h(\xi)\right)=0,
\end{align} with $\lambda=0$, if $\opnorm{\bvec{a}^\star}{q}^q<h(\xi)$, and $\lambda< 0$, if $\opnorm{\bvec{a}^\star}{q}^q=h(\xi)$. Now $\frac{\partial \mathcal{L}}{\partial a_j}|_{a_j=a_j^\star}=0\implies f'(a_j)=-\lambda q (a_j^\star)^{q-1},\ j=1,\cdots,s$. Now, if $\lambda=0$, clearly, we have $f'(a_j)=0,\ \forall j$, which is not possible as $f$ is strictly increasing in $[0,\infty)$. 
%Since $f'$ is a non-increasing function, we must have that either $a_j^\star=\infty$, or $a_j^\star\ge c$ for some $c>0$, such that $f'(c)=0$. However, the second case cannot occur since, by definition~\ref{def:f-properties}, $f$ is strictly increasing, implying that $f'(c)>0,\ \forall c>0$. Now, if $f$ is such that $f'(a)=0\implies a=\infty$, the corresponding solution vector $\bvec{a}_j^\star$ does not fall in $\mathcal{C}$.
Hence, one must have $\lambda< 0$, and $f'(a_j^\star)=-\lambda q (a_j^\star)^{q-1}\ \forall j=1,\cdots,s$. Now, as both the functions $f'(x)$ and $x^{1-q}$ are positive valued with the former non-increasing and the latter strictly decreasing for $x>0$ the function $f'(x)/x^{q-1}$ is monotonically decreasing and one-one. Hence, $\frac{f'(a_j^\star)}{(a_j^\star)^{q-1}}=-\lambda q\ \forall j=1,\cdots,s\implies a_1^\star=\cdots,\ a_s^\star$. Hence, from $\opnorm{\bvec{a}^\star}{q}^q=h(\xi)$, and using $h(\xi)=s\left(f^{-1}(\xi/s)\right)^q$ one obtains, $a_j^\star=f^{-1}(\xi/s),\ j=1,\cdots,s$. Therefore, $\max_{\bvec{a}\in \mathcal{C}}\sum_{j=1}^s f(a_j)=sf(a_1^\star)=\xi.$ Since $\xi>0$ was chosen arbitrarily, we have obtained the lower bound of~\eqref{eq:main-holder-type-inequality-generalized-concave-function}.
\subsection{Upper bound}
\label{sec:appendix-proof-of-lemma-generalized-holder-type-upper-bound}
\subsubsection{Proof sketch}
\label{sec:upper-bound-proof-sketch}
For proving the upper bound in~\eqref{eq:main-holder-type-inequality-generalized-concave-function}, we follow the approach taken in~\cite{gu2018sparse-lp-bound} for proving Lemma~$1$ therein. We first fix some $\xi >0$, and show that for some fixed $\alpha,\beta\ge 0$, \begin{align}
\alpha f^{-1}\left(\sum_{j=1}^s \frac{f(a_j)}{s}\right)+\beta \opnorm{\bvec{a}}{\infty}<\xi\implies \frac{\opnorm{\bvec{a}}{q}}{s^{1/q}}\le h(\xi),
\end{align} for some function $h$. Later we choose $\alpha,\beta$ in such a way that $h(\xi)=\xi,\ \forall\xi>0$, so that the desired claim is proved.

Let $\kappa(\bvec{a}) = \alpha f^{-1}\left(\sum_{j=1}^s f(a_j)/s\right)+\beta \opnorm{\bvec{a}}{\infty}$. We fix some $\xi>0$, and w.l.o.g. assume that $\bvec{a}\in \real_+^s$, and $a_1\ge a_2\ge \cdots\ge a_s\ge 0$. Consider the set $\mathcal{D}=\left\{\bvec{a}\in \real_+^s:a_1\ge a_2\ge \cdots a_s,\ \kappa(\bvec{a})\le \xi\right\}$. We want to maximize the function $\opnorm{\bvec{a}}{q}/s^{1/q}$ over the set $\mathcal{D}$. In order to do this, following~\cite{gu2018sparse-lp-bound}, we will create $\mathcal{D}_{1}$ which is the convex polytope with its vertices (extreme points) taken from a set of points in $\mathcal{D}$. Since $\opnorm{\cdot}{q}$ is convex for $q\ge 1$, the maximum value of $\opnorm{\bvec{a}}{q}/s^{1/q}$ over $\mathcal{D}_{1}$ is attained at one of these vertices. Finally, we will show that $\mathcal{D}_{1}\supset\mathcal{D}$, so that the maximum of $\opnorm{\bvec{a}}{q}/s^{1/q}$ over points of the set $\mathcal{D}$ is upper bounded by the maximum value obtained over the vertices of $\mathcal{D}_1$. We now proceed to find these extreme points.

\subsubsection{Finding the vertices}
\label{sec:upper-bound-proof-finding-vertices}
As in the proof of Lemma~$1$ in~\cite{gu2018sparse-lp-bound}, the $s+1$ vertices of the $\mathcal{D}_1$ are obtained by taking the intersection of the hypersurface $\kappa(\bvec{a})=\xi$ with hyperplanes of the form $a_1=a_2=\cdots=a_k>a_{k+1}=\cdots=a_s=0$, for $k=1,\cdots,s$, as well as considering the trivial solution $a_1=\cdots=a_s=0$. For a fixed $k=1,\cdots,s$, $a_k$ will satisfy \begin{align}
\label{eq:ak-equation}
 \alpha f^{-1}\left(\frac{k}{s}f(a_k)\right)+\beta a_k =\xi
 & \implies \gamma_kf(a_k) =f\left(\frac{\xi-\beta a_k}{\alpha}\right),
\end{align} where $\gamma_k=k/s,\ 1\le k\le s$. Hence the $s+1$ vertices will form the set $\mathcal{P}=\{\bvec{0},\bvec{p}_1,\cdots,\bvec{p}_s\}$, where $\bvec{p}_k=[\underbrace{a_k,\cdots,a_k}_{k\ \mathrm{times}},0,\cdots,0]^t,\ 1\le k\le s$.
At any of the vertices $\bvec{p}_k,\ 1\le k\le s$, the value of the function $\opnorm{\bvec{a}}{q}/s^{1/q}$ is $\gamma_k^{1/q}a_k$. We need to find $\max_{1\le k\le s}\gamma_k^{1/q}a_k$. Note that our goal is to choose $\alpha,\beta$ in such a way that this maximum value is $\le \xi$. Hence, we have to show that \begin{align}
\label{eq:desired-a_k-inequality}
\max_{1\le k\le s}\gamma_k^{1/q}a_k & \le \xi.
\end{align}  
For $k=s$, desired inequality~\eqref{eq:desired-a_k-inequality} reduces to $a_s\le \xi$. Since $a_s=\frac{\xi}{\alpha+\beta}$ (from Eq~\eqref{eq:ak-equation}),~\eqref{eq:desired-a_k-inequality} holds for $k=s$ if $\alpha+\beta\ge 1$. Now, let us fix some $k,\ 1\le k\le s-1$. If $\beta\ge \gamma_k^{1/q}$, $\gamma_k^{1/q}a_k\le \gamma_k^{1/q}\xi/\beta\le \xi$. If $\beta\le \gamma_k^{1/q}$, then note that for the inequality~\eqref{eq:desired-a_k-inequality} to hold, the following must be satisfied:\begin{align}
\label{eq:desired-a_k-inequality-required-condition-1_le_k_le_s-1}
\gamma_k f\left(\xi\gamma_k^{-1/q}\right) & \ge f\left(\frac{\xi-\beta \xi\gamma_k^{-1/q}}{\alpha}\right),
\end{align} 
which follows since the function $\gamma_kf(a_k)-f\left(\frac{\xi-\beta a_k}{\alpha}\right)$ is increasing in $a_k$, as long as $a_k\le \xi/\beta$. Since~\eqref{eq:desired-a_k-inequality-required-condition-1_le_k_le_s-1} has to be satisfied $\forall \xi>0$, one must have, along with $0\le \beta\le \gamma_k^{1/q},\ \alpha+\beta\ge 1$, \begin{align}
\label{eq:aq-maximization-problem-wrt-xi}
\sup_{\xi>0}\frac{f\left(\frac{\xi-\beta \xi\left(\gamma_k\right)^{-1/q}}{\alpha}\right)}{f\left(\xi\gamma_k^{-1/q}\right) } \le \gamma_k & \Leftrightarrow g\left(\frac{\gamma_k^{1/q}-\beta}{\alpha}\right) \le \gamma_k.
\end{align} 
Therefore, to choose $\alpha,\ \beta>0$ in such a way the inequality~\eqref{eq:desired-a_k-inequality} holds true, one needs to consider two cases:\begin{itemize}
	\item If $\beta\ge \left(1-1/s\right)^{1/q}$, the inequality~\eqref{eq:desired-a_k-inequality} holds true by default;
	\item If $0<\beta\le (l/s)^{1/q}$, for some $l=1,\cdots,s-1$, then, the inequality~\eqref{eq:aq-maximization-problem-wrt-xi} has to hold for $l\le k\le s-1$. 
\end{itemize} 
While the former is a valid choice for $\beta$, it is dependent upon $s$, and for large $s$ this choice might result in choosing $\beta$ to be too close to $1$. However, as it is obvious that the desired upper bound~\eqref{eq:main-holder-type-inequality-generalized-concave-function} in Lemma~\ref{lem:generalized-holder-type-bounds} is valid for any $\beta\ge 1$, one should consider if a choice of $\beta$ which is smaller than $1$ and independent of $s$ is possible. For this reason, one considers the latter case, which is satisfied if the following stronger inequality holds true: \begin{align}
\label{eq:aq-maximization-problem-wrt-xi-stronger-condition}
 g\left(\frac{x^{1/q}-\beta}{\alpha}\right) & \le x,\ \forall x\in (\beta^q,(\alpha u_0+\beta)^q),
\end{align} 
where $u_0=\inf\{u\in(0,1]:g(u)=1\}$. Clearly, this second choice of $\beta$ can be considered only when $u_0>0$. By a change of variables, $x\to u$, with $u=(x^{1/q}-\beta)/\alpha$, one can see that the above is equivalent to \begin{align}
\label{eq:aq-maximization-problem-wrt-xi-stronger-equivalent-condition}
\sup_{u\in(0,u_0)}\frac{g\left(u\right)}{(\alpha u+\beta)^q} & \le 1.
\end{align} 
Hence, we find that the inequality~\eqref{eq:desired-a_k-inequality} is satisfied if the conditions on $\alpha,\beta$ hold true as stated in Lemma~\ref{lem:generalized-holder-type-bounds}. 
Now we proceed to create the polytope $\mathcal{D}_{1}=\conv{\mathcal{P}}$.
\subsubsection{Proving that $\mathcal{D}\subset \mathcal{D}_1$}
\label{sec:proving-convex-polytope-is-larger}
We first fix any $\bvec{v}\in \mathcal{D}$. Note that the vectors $\{\bvec{p}_1,\cdots,\bvec{p}_s\}$ form a basis for $\real^s$, so that $\exists \bvec{c}\in \real^s$ such that $\bvec{v}=\sum_{i=1}^s c_i \bvec{p}_i$. Then, to prove that $\mathcal{D}\subset \mathcal{D}_1$, it is enough to show that $c_i\ge 0,\ i=1,\cdots,s$, and $\sum_{i=1}^s c_i\le 1$.

First, observe that, $v_j=\sum_{t=j}^sc_ta_t,\ 1\le j\le s$. Then, for any $1\le k\le s-1,$ \begin{align}
c_k=\frac{v_k-v_{k+1}}{a_k}, & c_s=\frac{v_s}{a_s}.
\end{align} Now note that $a_k>0$ for all $1\le k\le s$. Otherwise, if $a_k=0$, for some $1\le k\le s$, then from Eq~\eqref{eq:ak-equation}, $f(\xi/\alpha)=0\implies \xi=0$, which contradicts the assumption $\xi>0$. Now as $\bvec{v}\in \mathcal{D}$ implies that $v_1\ge \cdots v_s\ge 0$, we have that $c_k\ge 0$ for all $1\le k\le s$. Now, we need to show that $\sum_{k=1}^s c_k\le 1$, or equivalently,
\begin{align}
\sum_{k=1}^{s-1}\frac{v_{k}-v_{k+1}}{a_k}+\frac{v_s}{a_s} \le 1 &  \Leftrightarrow \frac{v_1}{a_1}+\sum_{k=1}^{s-1}v_{k+1}\left(\frac{1}{a_{k+1}}-\frac{1}{a_k}\right) \le 1.
\end{align} To prove this, let us define, for any $p_2\ge p_1\ge 0$, the set,
%\begin{align}
%\label{eq:increasing-coordinate-I(a)-set-defintiion}
%I(\bvec{a}) & = \left\{\bvec{v}\in \real^s:v_1\ge v_2\ge \cdots\ge v_s\ge 0\right\},
%\end{align} and  
\begin{align}
\label{eq:s-gamma-set-defintion}
S_{p_1,p_2} = & \left\{\bvec{v}\in \real^s:v_1\ge v_2\ge \cdots\ge v_s\ge 0,\ p_1\le \bvec{v}^t \bvec{w}\le p_2\right\},
\end{align} 
where \begin{align}
\bvec{w}=\left[\frac{1}{a_1}\ \frac{1}{a_2}-\frac{1}{a_1}\cdots\ \frac{1}{a_s}-\frac{1}{a_{s-1}}\right]^t. 
\end{align} Since $\mathcal{D}\subset \left\{\bvec{v}\in \real^s:v_1\ge v_2\ge \cdots\ge v_s\ge 0\right\}$, by definition, we need to prove that $\mathcal{D}\subset S^C$, where $S=\cup_{p_1,p_2:1< p_1\le p_2}S_{p_1,p_2}$. Let us assume that $\mathcal{D}\not\subset S^C$, so that $\exists \bvec{v}\in \mathcal{D}$, and $p_2\ge p_1>1$ such that $\bvec{v}\in S_{p_1,p_2}$ (Note that such $p_2$ must be finite since, $\bvec{v}^t\bvec{w}<v_1/a_s<\xi/(\beta a_s)$). Consider the function \begin{align}
h(\bvec{v}) & =\frac{1}{s}\sum_{k=1}^s f(v_k)-f\left(\frac{\xi-\beta v_1}{\alpha}\right).
\end{align} Since $\bvec{v}\in \mathcal{D}$, we have $h(\bvec{v})\le 0$. However, we will show that $h(\bvec{v})>0$ if $\bvec{v}\in S_{p_1,p_2}$ whenever $p_2\ge p_1>1$, so that we will arrive at a contradiction implying that $\bvec{v}\in S^C$. In order to do this, we will show that if $\bvec{v}\in S_{p_1,p_2}$ for some $p_1>1$, $\frac{1}{s}\sum_{k=1}^s f(v_k)>\eta$, whenever $f\left(\frac{\xi-\beta v_1}{\alpha}\right)\ge \eta,\ \forall \eta>0$. Hence, we need to show that the minimum value of the objective function of the following optimization problem is $>\eta$:\begin{align}
\label{eq:f-optimization-intermediate}
\min_{\bvec{v}} \frac{1}{s}\sum_{k=1}^s f(v_k),\ &
\mathrm{s.t.} \bvec{v}\in T_\eta,
\end{align}
where, \begin{align}
T_\eta=\{\bvec{v}\in \real^s:\rho(\eta)\ge v_1\ge \cdots\ge v_s\ge 0,p_1\le \bvec{v}^t \bvec{w}\le p_2\},
\end{align} where $\rho(\eta)=(\xi-\alpha f^{-1}(\eta))/\beta$. It is easy to check that the set $T_\eta$ is convex and obviously is a subset of $\real_+^s$. Since the function $\frac{1}{s}\sum_{k=1}^s f(v_k)$ is concave in $\real_+^s$, the minimum value is achieved at one of the extreme points of the set $T_\eta$. 

In order to find the extreme points of $T_\eta$, we use the Lemma~\ref{lem:verification-vertex-polyhedron}. First, we express $T_\eta$ in the canonical form of a polyhedra created by a set of linear inequalities as in Lemma~\ref{lem:verification-vertex-polyhedron}. The set $T_\eta$ can be expressed as $T_\eta=\{\bvec{x}\in\real^s:\bvec{A}\bvec{x}\le \bvec{b}\}$. Here $\bvec{A}\in \real^{(s+3)\times s}$, with the $i^\mathrm{th}\ (1\le i\le s+3)$ row of $\bvec{A}$ equal to $\bvec{a}_i^t$, 
% \begin{align}
%\bvec{A} = \begin{bmatrix}
%\bvec{a}_1^t\\
%\bvec{a}_2^t\\
%\vdots\\
%\bvec{a}_{s+3}^t
%\end{bmatrix},
%\end{align}
where $\bvec{a}_1=[1\ 0\ 0\ \cdots\ 0]^t$, $\bvec{a}_j=[0\ 0\ \cdots\ -1\ \underset{\underset{j^\mathrm{th}\ \text{position}}{\uparrow}}{1}\cdots\ 0],\ j=2,\cdots,s$, $\bvec{a}_{s+1}=[0\ 0\ 0\ \cdots\ 0 -1]^t$, $\bvec{a}_{s+2}=-\bvec{a}_{s+3}=\bvec{w}$, and $b_1=\rho(\eta),\ b_j=0,\ j=2,\cdots, s+1,\ b_{s+2}=p_2,\ b_{s+3}=-p_1$. 

Now consider an extreme point $\bvec{v}$ of $T_{\eta}$ such that $\bvec{v}$ has $l(\ge 1)$ distinct positive values. If $l=1$, all the indices have the same value $v_1$, while, for $l\ge 2$, $\exists$ indices $i_1> i_2>\cdots> i_{l-1}$, such that $v_1=\cdots=v_{i_1}>v_{i_1+1}=\cdots=v_{i_2}>\cdots v_{i_{l-1}+1}=\cdots=v_s$. Now consider the set $A_v=\{\bvec{a}_j\in \real^s:\bvec{a}_j^t\bvec{v}=b_j,\ j\in\{1,2,\cdots,s+3\}\}$. According to Lemma~\ref{lem:verification-vertex-polyhedron}, for $\bvec{v}$ to be an extreme point of $T_\eta$, the set $A_v$ must contain $s$ linearly independent vectors. Using the definition of the vectors $\bvec{a}_j,\ j=1,2,\cdots,s+3$, one finds that $\bvec{a}_j^t\bvec{v}=b_j$ implies the following: \begin{align}
\label{eq:vertex-condition-first}
v_1 & = b_1,\ \mbox{if}\ j=1,\\
\label{eq:vertex-condition-second}
v_{j} & = v_{j-1},\ \mbox{if}\ j=2,\cdots,s,\\
\label{eq:vertex-codnition-last}
v_{s} & = 0,\ \mbox{if}\ j=s+1,\\
\label{eq:vertex-condition-p_1-p_2}
\bvec{a}_j^t\bvec{v} & = p_2,\ \mbox{resp.}\ p_1,\ \mbox{if}\ j=s+2\ \mbox{resp.}\ s+3.
\end{align}  
 Since there are $l$ distinct values of $v_j$'s, a maximum of $s-l$ of the constraints of the form~\eqref{eq:vertex-condition-second} can be satisfied, which means that a maximum of $s-l$ vectors of the form $\bvec{a}_j, j=2,\cdots, s$ can be there in $A_v$. Furthermore, only one of the constraints in~\eqref{eq:vertex-condition-p_1-p_2} can be satisfied. Now note that if $l\ge 4$, the set $A_v$ can have a maximum of $s-4$ of the vectors $\bvec{a}_j,\ j=2,\cdots,s$, which along with the vectors $\bvec{a}_j,\ j=1,s+1$ and $j=s+2$ or $s+3$, form a linear independent set of a maximum of $s-1$ vectors. This implies that the entries of a vertex cannot take more than $3$ distinct values. 
 
 Now, we consider the different cases that arise considering the satisfaction of the constraints~\eqref{eq:vertex-condition-first} and~\eqref{eq:vertex-codnition-last}. If~\eqref{eq:vertex-condition-first} is satisfied, then $1\le l\le 3$ so that the extreme point has the form $[\underbrace{b_1,b_1,\cdots,b_1}_{i\ \mathrm{times}},\underbrace{u_2,\cdots,u_2}_{j\ \mathrm{times}},\underbrace{0,\cdots,0}_{k\ \mathrm{times}}]^t$ where $i,j,k$ are arbitrary non-negative integers such that $1\le i\le s$, $0\le j+k\le s-i$, and $u_2(>0)$ is determined from either $\bvec{v}^t\bvec{w}=p_1$ or $\bvec{v}^t\bvec{w}=p_2$. Using the expression for the coordinates of $\bvec{w}$, it follows that $u_2=\frac{a_{i+j}(qa_i-\rho(\eta))}{a_i-a_{i+j}}$(In this case one must have $qa_i>\rho(\eta)$), where $q$ is either $p_1$ or $p_2$. If~\eqref{eq:vertex-condition-first} is not satisfied, then $1\le l\le 2$, where $l=2$ if~\eqref{eq:vertex-codnition-last} is satisfied, and $l=1$ if it is not. Thus, in this case, the extreme points are of the form $[\underbrace{u_1,\cdots,u_1}_{r\ \mathrm{times}},\underbrace{0,\cdots,0}_{t\ \mathrm{times}}]^t$, where $1\le r\le s$, and $0\le t\le s-r$. Here $u_1$ is obtained from $\bvec{v}^t\bvec{w}=q$ ($q=p_1$ or $p_2$), so that $u_1=qa_r$. 

Therefore, to summarize, an extreme point of $T_\eta$ is either of the forms $[\underbrace{b_1,b_1,\cdots,b_1}_{i\ \mathrm{times}},\underbrace{u_2,\cdots,u_2}_{j\ \mathrm{times}},\underbrace{0,\cdots,0}_{k\ \mathrm{times}}]^t$ or, $ [\underbrace{u_1,\cdots,u_1}_{r\ \mathrm{times}},\underbrace{0,\cdots,0}_{t\ \mathrm{times}}]^t$ where $i,j,k,r,t$ are arbitrary non-negative integers such that $1\le i,r\le s$, $0\le j+k\le s-i,\ 0\le t\le s-r$, and $u_1,\ u_2$ are obtained as below: \begin{align}
\label{eq:vertex-u_1-expression}
u_1 & = q_1 a_r,\\
\label{eq:vertex-u_2-expression}
u_2 & = \frac{a_{i+j}(q_2a_i-\rho(\eta))}{a_i-a_{i+j}},
\end{align}
where $q_1,q_2\in \{p_1,p_2\}$. We will now show that the value of $\frac{1}{s}\sum_{i=1s}^s f(v_i)$ at each of these extreme points is greater than $\eta$.

If the extreme point is of the form $ [\underbrace{u_1,\cdots,u_1}_{r\ \mathrm{times}},\underbrace{0,\cdots,0}_{t\ \mathrm{times}}]^t$ with $u_1$ given by Eq~\eqref{eq:vertex-u_1-expression}, then the value of the objective function of the problem~\eqref{eq:f-optimization-intermediate} at this point is $\frac{r}{s}f(q_1a_r)$. Since the function $r/sf(a)-f\left((\xi-\beta a)/\alpha\right)$ is a strictly increasing function in $a$ as long as $a<\xi/\beta$, and since $v_1\le \rho(\eta)<\xi/\beta\implies q_1a_r<\xi/\beta$, one obtains, along with the fact that $q_1>1$, \begin{align}
\ & \frac{r}{s}f(q_1a_r)-f\left(\frac{\xi-\beta q_1a_r}{\alpha}\right) > \frac{r}{s}f(a_r)-f\left(\frac{\xi-\beta a_r}{\alpha}\right)=0\nonumber\\
\ & \implies \frac{r}{s}f(q_1a_r) > f\left(\frac{\xi-\beta q_1a_r}{\alpha}\right)\ge f\left(\frac{\xi-\beta \rho(\eta)}{\alpha}\right)=\eta.
\end{align}

On the other hand, if an extreme point is of the form $[\underbrace{b_1,b_1,\cdots,b_1}_{i\ \mathrm{times}},\underbrace{u_2,\cdots,u_2}_{j\ \mathrm{times}},\underbrace{0,\cdots,0}_{k\ \mathrm{times}}]^t$, the value of the objective function at such a point becomes $\displaystyle \frac{i}{s}f(b_1)+\frac{j}{s}f\left(\frac{a_{i+j}(q_2a_i-\rho(\eta))}{a_i-a_{i+j}}\right)$. Now, observe that due to the constraint $1<p_1\le \bvec{v}^t\bvec{w}\le p_2$ and $v_1\ge v_2\cdots\ge v_s$, one obtains, \begin{align}
\ & 1<\bvec{v}^t\bvec{w}=\frac{v_1}{a_1}+\sum_{k=1}^{s-1}v_{k+1}\left(\frac{1}{a_{k+1}}-\frac{1}{a_k}\right)\le \frac{v_1}{a_s}\nonumber\\
\ & \implies v_1>a_s\implies \rho(\eta)>a_s\nonumber\\
\ & \implies (\xi-\alpha f^{-1}(\eta))/\beta>\xi/(\alpha+\beta)\implies f(a_s)>\eta.
\end{align} Therefore, $f(b_1)=f(\rho(\eta))>f(a_s)>\eta$. On the other hand, \begin{align}
\lefteqn{ \frac{a_{i+j}(q_2a_i-\rho(\eta))}{a_i-a_{i+j}}-f^{-1}(\eta)} & &\nonumber\\
\ & =\frac{a_{i+j}(q_2\beta a_i-(\xi-\alpha f^{-1}(\eta)))-\beta f^{-1}(\eta)(a_i-a_{i+j})}{\beta(a_i-a_{i+j})}\nonumber\\
\ & =\frac{a_{i+j}(q_2\beta a_i-\xi)-f^{-1}(\eta)(\beta a_i-a_{i+j}(\beta+\alpha))}{\beta(a_i-a_{i+j})}\nonumber\\
\ & = \frac{a_{i+j}}{a_s(a_{i}-a_{i+j})}\left[a_s\left(q_2a_i-\frac{\xi}{\beta}\right) -f^{-1}(\eta)\left(\frac{a_ia_s}{a_{i+j}}-\frac{\xi}{\beta} \right)\right]\nonumber\\
\ & > \frac{a_{i+j}a_if^{-1}(\eta)}{a_s(a_{i}-a_{i+j})}\left(q_2-\frac{a_s}{a_{i+j}}\right)>0,
\end{align} 
where the last step used $a_s>f^{-1}(\eta)$ as well as the fact that $q_2>1\ge a_s/a_k,\ \forall k=1,2,\cdots,s$. Hence the value of the objective function at this extreme point is also greater than $\eta$. 

Consequently, we have shown that $h(\bvec{v})>0$, if $\bvec{v}\in S_{p_1,p_2}$, for some $1\le p_1\le p_2$. Hence, we must have that $\bvec{v}^t\bvec{w}\le 1,\ \forall\ \bvec{v}\in \mathcal{D}$. Therefore, we have proven that $\mathcal{D}\subset\mathcal{D}_1$.

%  Note that the function $h$ is concave, and hence it attains its minimum value over the set $S_{p_1,p_2}$, at one of its vertices. The vertices of the set $S_{p_1,p_2}$ are found by finding the points of intersections of the cone $v_1\ge v_2\ge \cdots\ge v_s\ge 0$ with the planes $\bvec{v}^t\bvec{w}=p_1$, and $\bvec{v}^t \bvec{w}=p_2$. Then, the vertices can be easily verified to be given by $\{\bvec{q}_k\}_{k=1}^s$, and $\{\bvec{r}\}_{k=1}^s$, where  $\bvec{q}_k=p_1 \bvec{p}_k$, and $\bvec{r}_k=p_2\bvec{p}_k$. Now, for any $\gamma>1$, and for any $1\le k\le s$, $h(\gamma\bvec{p}_k)=f^{-1}\left(\frac{k}{s}f(\gamma a_k)\right)-\frac{\xi-\beta \gamma a_k}{\alpha}$. Now, note that the function $f^{-1}\left(\frac{k}{s}f(a)\right)-\frac{\xi-\beta a}{\alpha}$ is an increasing function of $a$. Hence, as $\gamma>1$, $f^{-1}\left(\frac{k}{s}f(\gamma a_k)\right)-\frac{\xi-\beta \gamma a_k}{\alpha}>f^{-1}\left(\frac{k}{s}f(a_k)\right)-\frac{\xi-\beta a_k}{\alpha}=0$, which follows from Eq~\eqref{eq:ak-equation}. Thus, for any $1\le k\le s$, $h(\bvec{q}_k)>0$, and $h(\bvec{r}_k)>0$. Therefore, $h(\bvec{v})\ge 0,\ \forall \bvec{v}\in S_{p_1,p_2}$. However, $\forall\ \bvec{v}\in \mathcal{D}, h(\bvec{v})\le 0$. Hence, we must have that $\mathcal{D}\subset  S^C$. Thus, $\sum_{k=1}^s c_k\le 1$, and consequently, $\mathcal{D}\subset\mathcal{D}_1$.  
\subsubsection{Proving the upper bound} 
\label{sec:proving-upper-bound-lemma-main-inequality}
We are now in a position to prove the upper bound in~\eqref{eq:J-NSC-inequality}. As we have proved that $\mathcal{D}\subset \mathcal{D}_1$, we have $\max_{\bvec{a}\in \mathcal{D}}\opnorm{\bvec{a}}{q}/s^{1/q}\le \max_{\bvec{a}\in \mathcal{D}_1}\opnorm{\bvec{a}}{q}/s^{1/q}$. Since $\opnorm{\bvec{a}}{q}/s^{1/q}$ is a convex function (since $q\ge 1$), we have, using the selection criterion of $\alpha,\beta$ as stated in Lemma~\ref{lem:generalized-holder-type-bounds}, $\max_{\bvec{a}\in \mathcal{D}_1}\opnorm{\bvec{a}}{q}/s^{1/q}=\max_{1\le k\le s}(k/s)^{1/q}a_k/s^{1/q}\le \xi$. Since the selection criterion of $\alpha,\beta$ is independent of $\xi$, the upper bound of~\eqref{eq:main-holder-type-inequality-generalized-concave-function} follows. 
\section{Proving the Equations~\eqref{eq:g(y)-evluation-continuous-mixed-norm} and~\eqref{eq:psi(y)-evluation-continuous-mixed-norm}}
\label{sec:appendix-proving-equations-g(y)-continuos-mixed-norm}
We will prove the following more general result: for any continuous and bounded function $h:[0,1]\to \real$, \begin{align}
\label{eq:general-result-evluation-continuous-mixed-norm}	
\lim_{x\downarrow 0}\frac{\expect{p}{x^ph(p)}}{\expect{p}{x^p}} & = h(p_1),\	\lim_{x\uparrow \infty}\frac{\expect{p}{x^ph(p)}}{\expect{p}{x^p}} = h(p_2).
\end{align}
Using $h(p)=y^p$, which is continuous in $p\in[0,1]$ and bounded in $[\min\{y,1\},\ \max\{y,1\}]$, results in Eq.~\eqref{eq:g(y)-evluation-continuous-mixed-norm}, whereas, using $h(p)=p$, which is continuous in $p\in[0,1]$ and bounded in $[0,1]$, results in Eq.~\eqref{eq:psi(y)-evluation-continuous-mixed-norm}.\\
\textbf{Evaluation of $\lim_{x\downarrow 0}\frac{\expect{p}{x^ph(p)}}{\expect{p}{x^p}},\ y>0$:} Since the function $\frac{\expect{p}{x^ph(p)}}{\expect{p}{x^p}}$ is continuous in $x$ for $x>0$, the limit to be evaluated is simply equal to $\lim_{n\to \infty}\frac{\expect{p}{x_n^ph(p)}}{\expect{p}{x_n^p}}$ for any sequence $\{x_n\}_{n\ge 0}$, such that $x_n> 0\forall n\ge 0$, and $x_n\stackrel{n\to \infty}{\to} 0$. Let us choose arbitrarily such a sequence $\{x_n\}_{n\ge 0}$. In order to evaluate $\lim_{n\to \infty}\frac{\expect{p}{x_n^ph(p)}}{\expect{p}{x_n^p}}$, we first observe that one can write $\frac{\expect{p}{x_n^ph(p)}}{\expect{p}{x_n^p}}=\int h(p)d\mu_n$, where $\mu_n$ is a probability measure defined on the Borel sets associated to the set $[0,1]$, such that, for each Borel set $A$ (defined on $[0,1]$), $\mu_n(A)=\frac{\int x_n^p\indicator{A}d\nu}{\expect{p}{x_n^p}}$. The cumulative distribution function (cdf) corresponding to $\mu_n$ is $F_n$, where, for any $u\in [0,1]$, $F_n(u)=\mu_n([0,u])=\frac{\int x_n^p\indicator{[0,u]}d\nu}{\expect{p}{x_n^p}}$. We will now prove that $F_n(u)\to \indicator{{u\ge p_1}}$, as $n\to \infty$, $\forall u\ne p_1$, where $p_1=\inf \Omega$, where $\Omega$ is the support of $\nu$ as well as $\mu_n,\ \forall n\ge 0$. Since the set $\real\setminus \{p_1\}$ is the set of continuity of the function $\indicator{{u\ge p_1}}$, this will imply, by the equivalence between weak convergence of distributions and weak convergence of corresponding probability measures (see Theorem 13.23 of~\cite{klenke2013probability}) that $\int h(p) d\mu_n\stackrel{n\to \infty}{\to} \int h(p) d\mu$, for any continuous and bounded function $h:[0,1]\to \real$, where $\mu$ is the measure associated to the cdf $\indicator{{u\ge p_1}}$, i.e., $\mu$ is the discrete probability measure satisfying $\mu(\{p_1\})=1$. Thus it remains to prove that $F_n(u)\stackrel{n\to \infty}{\to}\indicator{u\ge p_1}$, for any $u\ne p_1$. 

First of all, as $p_1=\inf \Omega$, obviously $F_n(u)=0,\ \forall u\in [0,p_1)$. Similarly if $u> p_2=\sup \Omega$, obviously $F_n(u)=1$. Now, fix any arbitrary $u\in(p_1,p_2]$ and $\epsilon\in (0,1)$. If we can prove that $\exists $ a positive integer $N$, such that $\forall n\ge N$, \begin{align}
\label{eq:limit-x-downto-0-limit-objective}
\abs{\frac{\int x_n^p\indicator{[p_1,u]}d\nu}{\expect{p}{x_n^p}}-1} \le \epsilon & \Leftrightarrow \frac{\int_{[p_1,u)} x_n^pd\nu}{\int_{[u,p_2]} x_n^pd\nu} \ge \frac{1}{\epsilon}-1,
%\Leftrightarrow \abs{\frac{\int x_n^p\indicator{[u,p_2]}d\nu}{\int x_n^pd\nu}} & \le \epsilon,\nonumber\\
%\Leftrightarrow \left[1+\frac{\int_{[p_1,u)} x_n^pd\nu}{\int_{[u,p_2]} x_n^pd\nu}\right]^{-1} & \le \epsilon\nonumber\\
\end{align} we are done, since $\epsilon$ was chosen arbitrarily. 
Now, as $1>p_2\ge u>p_1,\ \exists \epsilon_1>0$, such that $p_1+\epsilon_1<u$. Now, as $x_n\stackrel{n\to \infty}{\to}0$, $\exists$ positive integer $N_1$, such that $x_n\le 1,\ \forall n\ge N_1$. Then, $\forall n\ge N_1$, $\int_{[p_1,u)} x_n^pd\nu\ge x_n^{p_1+\epsilon_1}\nu([p_1,p_1+\epsilon_1])$, and $\int_{[u,p_2]} x_n^pd\nu\le x_n^{u}\nu([u,p_2])$, which imply that \begin{align}
\label{eq:continuous-mixed-norm-lower-limit-intermediate-step}
\frac{\int_{[p_1,u)} x_n^pd\nu}{\int_{[u,p_2]} x_n^pd\nu} & \ge x_n^{p_1+\epsilon_1-u}\frac{\nu([p_1,p_1+\epsilon_1])}{\nu([u,p_2])}.
\end{align} Let \begin{align}
\label{eq:continuous-mixed-norm-lower-limit-intermediate-step-eppsilon2-def}
\epsilon_2 & =\left[\left(\frac{1}{\epsilon}-1\right)^{-1}\frac{\nu([p_1,p_1+\epsilon_1])}{\nu([u,p_2])}\right]^{1/(u-(p_1+\epsilon_1))}.
\end{align} Again, since $x_n\stackrel{n\to \infty}{\to}0$, $\exists$ a positive integer $N_2$, such that $x_n\le \epsilon_2,\ \forall n\ge N_2$. Then, for $n\ge N:=\max\{N_1,N_2\}$, ~\eqref{eq:limit-x-downto-0-limit-objective} follows from~\eqref{eq:continuous-mixed-norm-lower-limit-intermediate-step} and~\eqref{eq:continuous-mixed-norm-lower-limit-intermediate-step-eppsilon2-def}.

\textbf{Evaluation of $\lim_{x\uparrow \infty}\frac{\expect{p}{x^ph(p)}}{\expect{p}{x^p}},\ y>0$:} This evaluation is largely similar to the evaluation of $\lim_{x\downarrow 0}\frac{\expect{p}{x^ph(p)}}{\expect{p}{x^p}},\ y>0$. Here again, the goal is to evaluate $\lim_{n\to \infty}\frac{\expect{p}{x_n^ph(p)}}{\expect{p}{x_n^p}}$ for any sequence of positive numbers $\{x_n\}$, such that $x_n\stackrel{n\to \infty}{\to}\infty$. We fix such an arbitrary such sequence $\{x_n\}_{n\ge 0}$. Again, one can write $\frac{\expect{p}{x_n^py^p}}{\expect{p}{x_n^p}}=\int y^{p}d\mu_n$, with the same description of $\mu_n,F_n$ as before. Again, using the Theorem~13.23 of~\cite{klenke2013probability} about the equivalence of the weak convergence of distributions and the weak convergence of probability measures, we will establish that $\lim_{x\uparrow \infty}\frac{\expect{p}{x^ph(p)}}{\expect{p}{x^p}}=h(p_2)$ with $p_2=\sup \Omega$, by showing that $\lim_{n\to \infty}F_n(u)=\indicator{u\ge p_2}$, $\forall u\ne p_2$. To establish this, first note that clearly $F_n(u)=1$, if $u>p_2$, and $F_n(u)=0$, if $u<p_1$. We thus need to show that, for any fixed $u\in [p_1,p_2)$, and arbitrarily chosen $\epsilon>0$, \begin{align}
\label{eq:limit-x-upto-infty-limit-objective}
 \abs{\frac{\int x_n^p\indicator{[p_1,u]}d\nu}{\int x_n^pd\nu}} \le \epsilon &\Leftrightarrow \frac{\int_{[u,p_2]} x_n^pd\nu}{\int_{[p_1,u]} x_n^pd\nu}  \ge \frac{1}{\epsilon}-1.
 %\Leftrightarrow \left[1+\frac{\int_{[u,p_2]} x_n^pd\nu}{\int_{[p_1,u]} x_n^pd\nu}\right]^{-1} & \le \epsilon\nonumber\\
\end{align}  
Note that since $p_1\le u<p_2$, $\exists \epsilon_1>0$, such that $u+\epsilon_1<p_2$. Now, as $x_n\stackrel{n\to \infty}{\to}\infty$, $\exists$ positive integer $N_1$, such that $\forall n\ge N_1$, $x_n\ge 1$. Thus, whenever $n\ge N_1$, $\int_{[u,p_2]} x_n^pd\nu\ge x_n^{p_2-\epsilon_1}\nu([p_2-\epsilon_1,p_2])$, and $\int_{[p_1,u]} x_n^pd\nu\le x_n^u \nu([p_1,u])$. Therefore, $\forall n\ge N_1$, \begin{align}
\label{eq:continuous-mixed-norm-upper-limit-intermediate-step}
\frac{\int_{[u,p_2]} x_n^pd\nu}{\int_{[p_1,u]} x_n^pd\nu} & \ge x_n^{p_2-\epsilon_1-u}\frac{\nu([p_2-\epsilon_1,p_2])}{\nu([p_1,u])}.
\end{align} Let \begin{align}
\label{eq:continuous-mixed-norm-upper-limit-intermediate-step-eppsilon2-def}
\epsilon_2 & = \left[\left(\frac{1}{\epsilon}-1\right)\frac{\nu([p_1,u])}{\nu([p_2-\epsilon_1,p_2])}\right]^{1/(p_2-\epsilon_1-u)}.
\end{align}
Now, as $x_n\stackrel{n\to \infty}{\to}\infty$, $\exists $ positive integer $N_2$, such that $\forall n\ge N_2$, $x_n\ge \epsilon_2$. Then, $\forall n\ge N:=\max \{N_1,N_2\} $, it follows from~\eqref{eq:continuous-mixed-norm-upper-limit-intermediate-step} and~\eqref {eq:continuous-mixed-norm-upper-limit-intermediate-step-eppsilon2-def} that ~\eqref{eq:limit-x-upto-infty-limit-objective} is satisfied.

\section{Proof of Theorem~\ref{thm:null-space-constant-upper-bound}}
\label{sec:appendix-proof-of-theorem-null-space-constant-upper-bound}
Our proof is an extension of the proof of Theorem $4$ in~\cite{gu2018sparse-lp-bound}. Before commencing the proof, we describe a standard support partitioning trick~\cite{candes_decoding_2005,foucart2013mathematical,gu2018sparse-lp-bound} that is later used in the proof: Any vector $\bvec{z}\in \mathcal{N}(\bvec{\Phi})$ is decomposed by partitioning the indices $1,2,\cdots,N$ into the sets $T_i,\ i=0,1,\cdots$, where:\\ 
$\bullet$ $T_0$ is the index set corresponding to the coordinates of $\bvec{z}$ having the largest $K$ absolute values;\\
$\bullet$ $T_1$ is the index set corresponding to the coordinates of $\bvec{z}_{\overline{T}_0}$ with the largest $2K_0-K$ absolute values;\\
$\bullet$ $T_i(i\ge 2)$ is the index set corresponding to the coordinates of $\bvec{z}_{\overline{{T}_0\cup T_1\cup\cdots \cup T_{i-1}}}$ with the largest $2K_0-K$ absolute values.   

Using this partition, our proof begins with Eq~($52$) and Eq~($53$) of~\cite{gu2018sparse-lp-bound}, which are reproduced respectively below: \begin{align}
\label{eq:gu-theorem4-proof-first-eqn}
\norm{\bvec{z}_{T_0}} & \le C'_{2K_0}\sum_{i\ge 2}\opnorm{\bvec{z}_{T_i}}{1},\\
\label{eq:gu-theorem4-proof-second-eqn}
\sum_{i\ge 2}\norm{\bvec{z}_{T_i}} & \le \frac{\opnorm{\bvec{z}_{T'}}{1}}{\sqrt{L}}+\omega z_2^+,
\end{align}
where $C_{2K_0}'=\frac{\sqrt{2}+1}{2}\frac{\delta_{2K_0}}{1-\delta_{2K_0}}$,\ $L=2K_0-K$, $z_i'=\opnorm{\bvec{z}_{T_i}}{\infty},\ i\ge 1$, $T'=\overline{T_0\cup T_1}\setminus \{2K_0+1\}$, and $\omega=\frac{\sqrt{L}}{2}+\frac{1}{\sqrt{L}}$. Let $\{T_i'\}_{i\ge 2}$ form a partition of the set $T'$ such that $T_i'=\{(L+1)(i-1)+K+1,\cdots,(L+1)i+K\}$, which has $L+1$ entries. Now, rewriting $\opnorm{\bvec{z}_{T'}}{1}$ as $\sum_{i\ge 2}\opnorm{\bvec{z}_{T_i'}}{1}$ and using Lemma~\ref{lem:generalized-holder-type-bounds} for function $f$, with $L'=L+1$, the RHS (Right Hand Side) of~\eqref{eq:gu-theorem4-proof-second-eqn} can be further upper bounded as below: \begin{align}
\sum_{i\ge 2}\norm{\bvec{z}_{T_i}} \le \omega z_2^+ + & \frac{L'}{\sqrt{L}}\sum_{i\ge 2}\left[\alpha_f(1,L')f^{-1}\left(\frac{F(\bvec{z}_{T_i'})}{L'}\right)+\beta_f(1,L')z_i'^+ \right],
\end{align}
where $z_i'^+=\opnorm{\bvec{z}_{T_i'}}{\infty}$. Now, using \begin{align*}
z_2'^+ \le z_2^+,\ z_{i+1}'^+ & \le f^{-1}\left(\frac{F(\bvec{z}_{T_i'})}{L'}\right),\ i\ge 1,\\ 
\sum_{i\ge 2} f^{-1}\left(\frac{F(\bvec{z}_{T_i'})}{L'}\right) & \le f^{-1}\left(\frac{\sum_{i\ge 2}F(\bvec{z}_{T_i'})}{L'}\right) =f^{-1}\left(\frac{F(\bvec{z}_{T'})}{L'}\right),
\end{align*}
one obtains, \begin{align}
\sum_{i\ge 2}\norm{\bvec{z}_{T_i}} & \le \frac{L'}{\sqrt{L}}(\alpha_f(1,L')+\beta_f(1,L'))f^{-1}\left(\frac{F(\bvec{z}_{T'})}{L'}\right) +\left(\omega + \frac{\beta_f(1,L')L'}{\sqrt{L}}\right)f^{-1}\left(\frac{F(\bvec{z}_{\overline{T_0}\setminus T'})}{L'}\right)\nonumber\\
\ & \le \frac{L'}{\sqrt{L}}(\pi_f(L')+\beta_f(1,L'))f^{-1}\left(\frac{F(\bvec{z}_{\overline{T_0}})}{L'}\right),
\end{align}
where $\pi_f(L')=\max\{\alpha_f(1,L'),\frac{1}{2}+\frac{1}{2L'}\}$. Furthermore, using the first inequality of~\eqref{eq:main-holder-type-inequality-generalized-concave-function} of Lemma~\ref{lem:generalized-holder-type-bounds} for $q=2$, along with inequality~\eqref{eq:gu-theorem4-proof-first-eqn}, one obtains, \begin{align}
f^{-1}\left(\frac{F(\bvec{z}_{T_0})}{K}\right) & \le \zeta_f(K,K_0) f^{-1}\left(\frac{F(\bvec{z}_{\overline{T_0}})}{L'}\right)\nonumber\\
\ \implies \frac{F(\bvec{z}_{T_0})}{K} & \le f\left(\zeta_f(K,K_0) f^{-1}\left(\frac{F(\bvec{z}_{\overline{T_0}})}{L'}\right)\right)\nonumber\\
\ \implies F(\bvec{z}_{T_0}) & \le \frac{K}{2K_0-K+1}g( \zeta_f(K,K_0))F(\bvec{z}_{\overline{T_0}}).
\end{align}
\section{Evaluating NSC when \texorpdfstring{$\mathcal{N}(\bvec{\Phi})=\spn{\bvec{z}}$}{blah}}
\label{sec:appendix-evaluating-nsc-null-space-spanned-by-single-vector}
In this case, $\bvec{u}\in \mathcal{N}(\bvec{\Phi})\implies \bvec{u}=t\bvec{z}$, for some $t\in \real$. We assume that the function $xf'(x)/f(x)$ is non-increasing, which is true for $f(x)=\abs{x}^p,\ln(1+\abs{x}^p),\ 1-e^{-\abs{x}^p}$, for $p\in(0,1]$.
%We have the following lemma: 

%\begin{lem}
%	\label{lem:nsc-for-dim-1-null-space}
%	If $\mathcal{N}(\bvec{\Phi})=\spn{\bvec{z}}$, and if the function $xf'(x)/f(x)$ is non-increasing, 
%\end{lem}
Writing $[N]=\{1,2,\cdots,N\}$, one then finds that \begin{align}
\gamma(J,\bvec{\Phi},K) & = \max_{S\subset [N]:\abs{S}\le K}\sup_{t>0}\frac{\sum_{i\in S}f(tz_i)}{\sum_{i\notin S}f(tz_i)}\nonumber\\
\ & = \left[\left(\max_{S\subset [N]:\abs{S}\le K}\sup_{t>0}\frac{\sum_{i\in S}f(t\abs{z_i})}{\sum_{i\in[N]}f(t\abs{z_i})}\right)^{-1}-1\right]^{-1}\nonumber\\
\ & = \left[\left(\sup_{t>0}\max_{S\subset [N]:\abs{S}\le K}\frac{\sum_{i\in S}f(t\abs{z_i})}{\sum_{i\in[N]}f(t\abs{z_i})}\right)^{-1}-1\right]^{-1}\nonumber\\
\ & = \left[\left(\sup_{t>0}\frac{\sum_{i=1}^Kf(tz_i^+)}{\sum_{i\in[N]}f(tz_i^+)}\right)^{-1}-1\right]^{-1}\nonumber\\
\ & = \sup_{t>0}\frac{\sum_{i=1}^Kf(tz_i^+)}{\sum_{i=K+1}^Nf(tz_i^+)},
\end{align} 
where $\bvec{z}^+$ is the positive rearrangement of $\bvec{z}$, i.e., for $i=1,\cdots,N$, $z_i^+$ is, magnitude-wise, the $i^\mathrm{th}$ largest entry of $\bvec{z}$.
Now, consider finding $\sup_{t>0}r(t)$, where $r(t) = \frac{\sum_{i=1}^Kf(tz_i)}{\sum_{i=K+1}^Nf(tz_i)},\ t>0$, for a vector $\bvec{z}$ with non-negative entries arranged as $z_1\ge z_2\ge \cdots\ge z_N\ge 0$. Then \begin{align}
\label{eq:r'(t)-expression}
\lefteqn{r'(t)} & &\nonumber\\
\ &  = \frac{\sum_{i=1}^K\sum_{j=K+1}^Nf(tz_i)f(tz_j)\left(\frac{t\abs{z_i}f'(t\abs{z_i})}{f(t\abs{z_i})} - \frac{t\abs{z_j}f'(t\abs{z_j})}{f(t\abs{z_j})}\right)}{t(\sum_{i\notin S}f(tz_i))^2}.
\end{align}  Therefore, as $xf'(x)/f(x)$ is non-increasing, one can easily check that $r(t)$ is a non-increasing function of $t$, so that $\sup_{t>0} r(t)=\lim_{t\downarrow 0}r(t).$ Consequently, $\gamma(J,\bvec{\Phi},K) = \lim_{t\downarrow 0}\frac{\sum_{i=1}^Kf(tz_i^+)}{\sum_{i=K+1}^Nf(tz_i^+)}.$ Therefore, for for $f(x)=\abs{x}^p,\ln(1+\abs{x}^p),1-e^{-\abs{x}^p},\ \gamma(J,\bvec{\Phi},K) = \frac{\sum_{i=1}^K\abs{z_i^+}^p}{\sum_{i=K+1}^N\abs{z_i^+}^p}$. 
\bibliography{null-space-J-minimization}
\end{document}